\documentclass{egpubl}
\usepackage{egsgp2026}
\usepackage{mathtools}
 
\SpecialIssuePaper         

\CGFccbyncnd

\usepackage[T1]{fontenc}
\usepackage{dfadobe}  
\usepackage{soul}

\usepackage{amsmath, amssymb, amsfonts}

\newtheorem{theorem}{Theorem}
\newtheorem{remark}{Remark}

\usepackage{booktabs} 

\usepackage{algorithm}
\usepackage{algorithmic}
\renewcommand{\algorithmicrequire}{\textbf{Input:}}
\renewcommand{\algorithmicensure}{\textbf{Output:}}

\usepackage{lineno}

\biberVersion
\BibtexOrBiblatex
\usepackage[backend=biber,bibstyle=EG,citestyle=alphabetic,backref=true]{biblatex} 
\electronicVersion
\PrintedOrElectronic

\ifpdf \usepackage[pdftex]{graphicx} \pdfcompresslevel=9
\else \usepackage[dvips]{graphicx} \fi

\usepackage{egweblnk} 

\usepackage[prologue,table,xcdraw]{xcolor}

\title[Tangent Blow-Ups]{Tangent Blow-Ups for Processing Non-Manifold Geometry}

\author[A. Petrov \& M. Nabizadeh \& A. Dodik \& J. Solomon]{
\parbox{\textwidth}{\centering
Alice~Petrov$^{1}$
\quad
Mohammad~Sina~Nabizadeh$^{1}$
\quad
Ana~Dodik$^{1}$
\quad
Justin~Solomon$^{1}$
\\
\parbox{\textwidth}{\centering
$^{1}$MIT CSAIL, Cambridge, USA
}}}

\begin{document}

\teaser{
 \includegraphics[width=\linewidth]{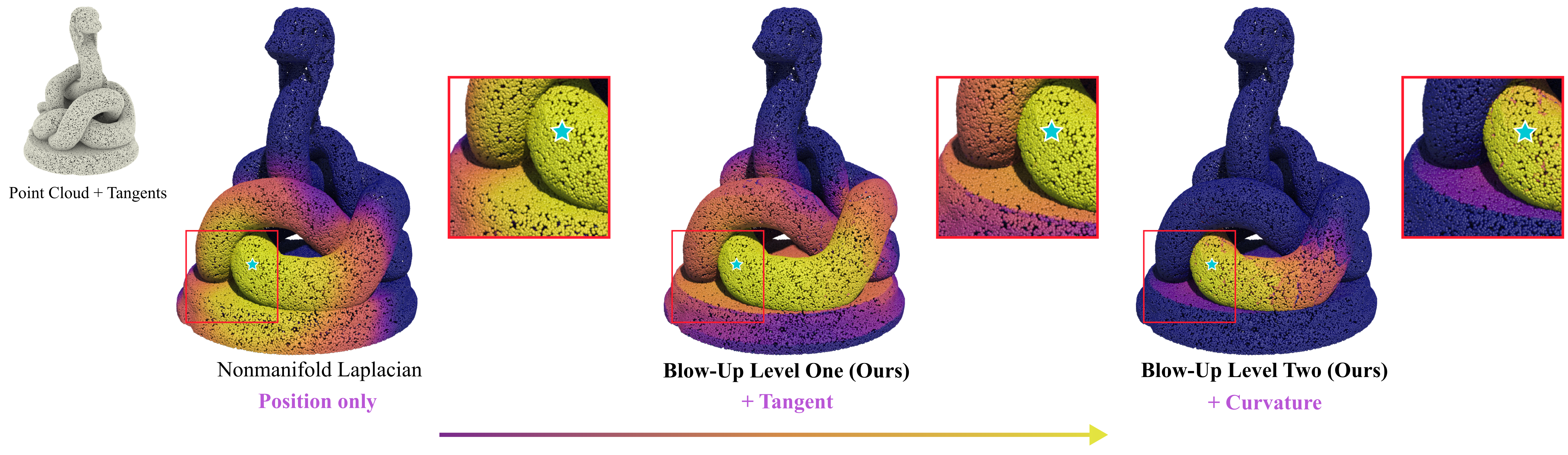}
 \centering
  \caption{In the figure, different metrics give different patterns of heat diffusion. At singularities and intersections, standard methods (left) \cite{sharp2020laplacian} treat all geometric components equally. Our tangent blow-up representation enriches each point with geometric information such as tangents (middle) and curvature (right) by iteratively lifting data into a higher-dimensional space where singularities are resolved and distances can better reflect the underlying geometry. This construction enables manifold-based point cloud processing pipelines to extend to non-manifold and singular shapes.}
\label{fig:teaser}
}

\maketitle

\begin{abstract}
Many geometry processing pipelines implicitly assume their input data is a manifold, or is sampled from one, with a unique tangent plane at every point. Geometric data, however, routinely contains sharp features like edges, corners, self-intersections, branching junctions, and other singularities, rendering standard methods ill-defined at these points. To bring geometry processing to these and other singular spaces, we introduce the ``tangent blow-up,'' a representation inspired by algebraic geometry that restores structure at singularities by lifting to the product of the ambient space and the Grassmannian of tangent planes. After iterating this construction, points that coincide in position but differ in tangent direction, curvature, or higher-order contact become well-separated. We equip the tangent blow-up with a product metric and define discretized differential operators, such as the gradient, divergence, and Laplacian, directly in the lifted domain. We demonstrate our framework across geodesic computation, segmentation, surface parameterization, and curvature estimation.
\end{abstract}



\section{Introduction}

The discretization of differential operators on point clouds typically relies on building local neighborhoods where the surface is well-approximated by a single tangent plane~\cite{hoppe1992surface, vaxman2016directional, robbin2022introduction}.
For example, the point cloud Laplacian constructs adjacency graphs from spatially-nearby points~\cite{coifman2006diffusion, belkin2009constructing}, and curvature estimators like polynomial jet fitting locally fit truncated Taylor expansions to the surface~\cite{CazalsPouget2005}.
The unifying assumption behind these methods is that data lies on or near a manifold.
Non-manifold structures, however, abound in real-world data.
For example, in geological modeling, the geometry and mechanical interactions 
at intersecting fault junctions dictate stress accumulation and slip distributions during seismic events~\cite{marshall2008effects}.

Standard geometry processing algorithms can be applied to non-manifold data, but their output at points with multiple intersecting tangent planes---so-called \emph{singular points}---typically does not respect the underlying geometry.
For example, in Figure~\ref{fig:teaser}, the nonmanifold Laplacian~\cite{sharp2020laplacian} diffuses heat across intersecting surfaces.
More generally, any algorithm that builds neighborhoods from Euclidean distance alone will inevitably mix points from geometrically unrelated components.
Hence, the question we address in this work is: \emph{how can differential operators be defined on point clouds whose underlying geometry contains singularities, without requiring singularity detection or manifold repair as a preprocessing step?}

Few computational approaches exist for directly handling non-manifold data.
Existing techniques either detect and resolve singularities during preprocessing, as in non-manifold mesh repair~\cite{charton2021mesh}, or estimate topological invariants while bypassing discretization of differential operators, as in persistent intersection homology~\cite{Bendich2011}.
One can attempt to stratify a given dataset into clusters, each of which is well-approximated by a suitable submanifold of the ambient space~\cite{Stolz2020, aamari2024theory}, a procedure known as \emph{stratification learning}~\cite{bendich2010stratification, Nanda2020}.
However, such decompositions can alter the topology of the space.
For example, the lemniscate (figure-eight) is a classic example of an immersed curve that is not an embedding, and decomposing it into two manifold pieces destroys the global connectivity of the self-intersecting loop (see Figure~\ref{fig:gfig8-decomposition}).

\begin{figure}[htb]
  \centering
  \includegraphics[width=\linewidth]{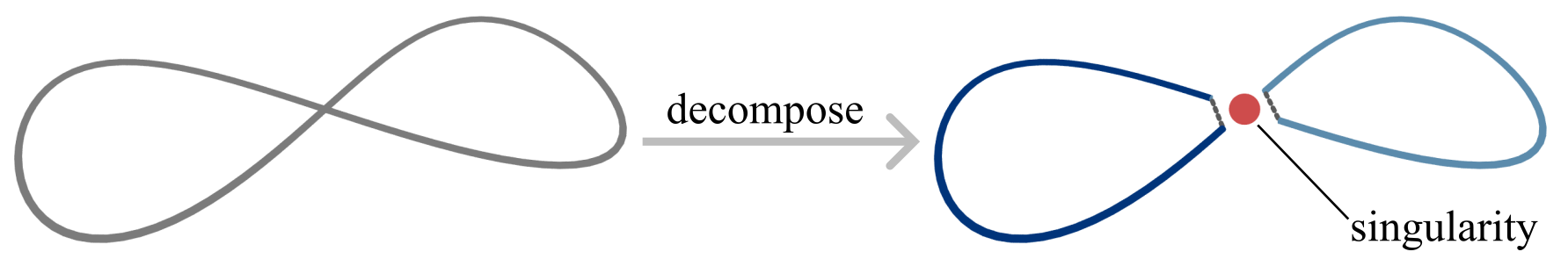}
  \caption{\label{fig:gfig8-decomposition}Decomposing a figure-eight curve into manifold pieces destroys the global connectivity of the self-intersecting loop.}
\end{figure}

The core difficulty at a singularity is that the tangent plane is not unique, and it is precisely the tangent plane that many algorithms use to define local neighborhoods, differential operators, and curvature. 
In the continuous realm, the \emph{Nash blow-up}~\cite{Nobile1975, Hironaka1983} in algebraic geometry resolves singularities in algebraic varieties by lifting points into the product space of positions and tangent planes (see Figure~\ref{fig:blowup-intuition}).  
In the resulting space, transversely intersecting points that coincide spatially are separated by their tangent-plane coordinates.
Unlike more heuristic approaches, the Nash blow-up resolves singularities by embedding points into regular manifolds in a higher-dimensional space.

We bring this continuous construction into the discrete, computational setting of point cloud geometry processing. We equip the product space of positions and tangent planes with a metric that admits an isometric embedding into Euclidean space. 
This embedding reduces the entire blow-up pipeline to a coordinate transformation, after which Euclidean constructions such as $k$-nearest-neighbor graphs and Gaussian kernels apply directly, enabling the spectral and curvature analyses shown in Figure~\ref{fig:mobius-intro}.
We further show that, since the output has the same structure as the input, the blow-up can be iterated. Each iteration 
resolves tangential intersections of progressively higher contact order (see Figure~\ref{fig:second-blowup}).
To justify the discrete pipeline, we analyze the continuous lifted space from which we assume points are sampled. We prove a \emph{separation theorem}: the lifted metric enforces a lower bound between points with distinct tangents
and a \emph{smoothness theorem}: the lifted embedding of a smooth manifold is itself smooth.
These results guarantee before sampling that transversely intersecting components are fully resolved into smooth, well-separated manifolds in the lifted space.
Since we sample point clouds from this lifted space, the separation bound transfers directly to the embedded points (Corollary~\ref{cor:discrete-separation}), and standard convergence results for point cloud Laplacians~\cite{belkin2006convergence,belkin2009constructing} apply to each manifold component individually. We address nontransverse intersections in \S\ref{sec:iterated-blowup}.

\begin{figure}
  \centering
  \includegraphics[width=\linewidth]{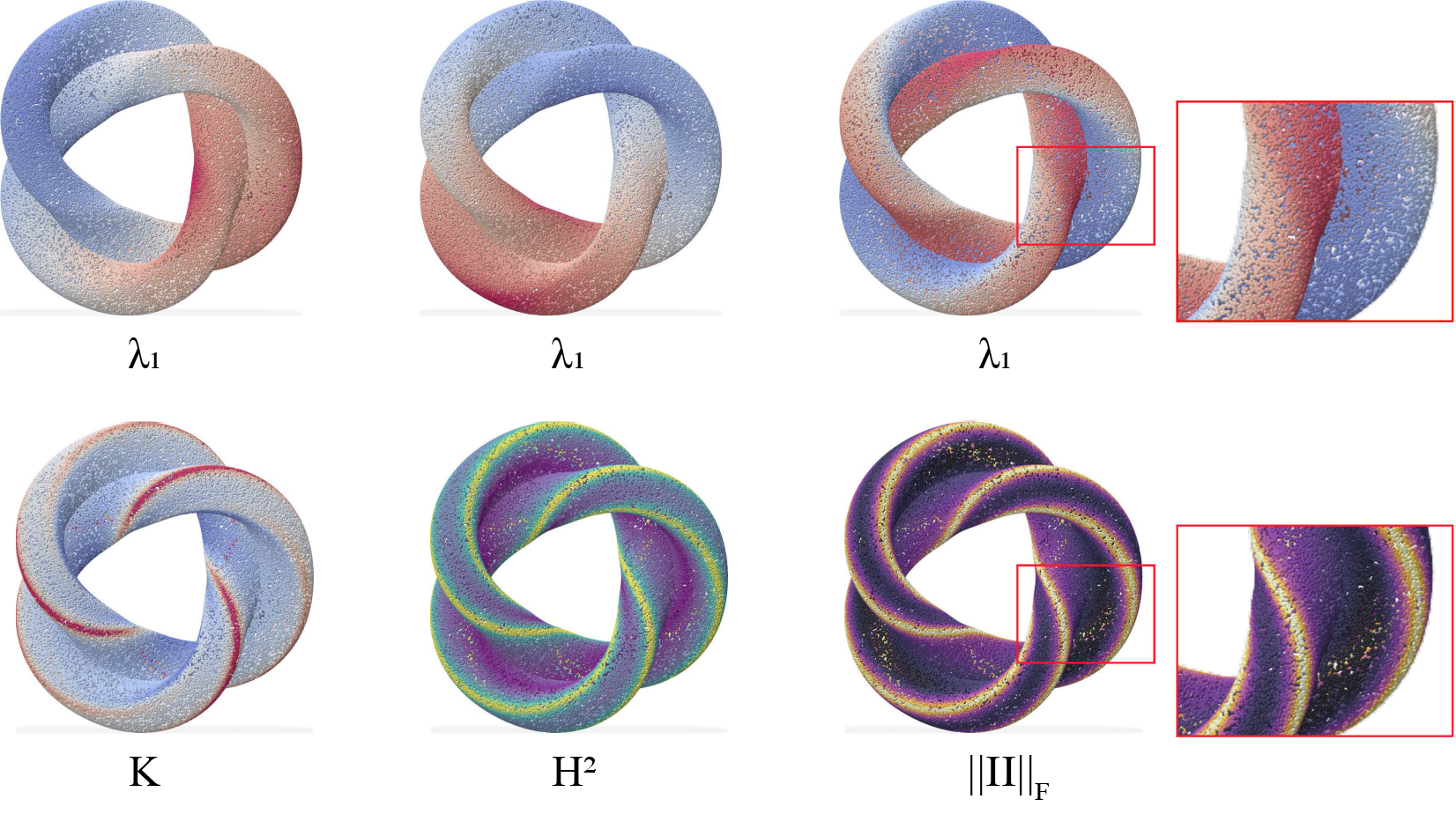}
  \caption{\label{fig:mobius-intro}Spectral and curvature analysis of Banchoff's Klein bottle, a self-intersecting nonorientable surface~\cite{banchoff1976minimal,Busser-nonorientable}. Top row: the three smallest non-trivial eigenvectors of our lifted Laplacian (\S\ref{sec:lifted-laplacian}). Bottom row: Gaussian curvature, squared mean curvature, and the Frobenius norm of the second fundamental form, all computed from the unoriented projector field (\S\ref{sec:curvature})}. 
\end{figure}

Our contributions are:
\begin{itemize}
\item A discrete tangent blow-up representation for lifting point clouds that enables discretized differential operators to be well-defined directly on non-manifold geometry without singularity detection or manifold repair.

\item An iterated blow-up that resolves tangential intersections of progressively higher contact order. We demonstrate that the second level separates components that agree to first order but differ in second-order geometry.

\item Separation and smoothness theorems for the continuous lifted space that justify the discrete construction. For transverse intersections, we prove a distance lower bound proportional to the angular gap between tangent planes and show that the lifted embedding of a smooth manifold is also smooth. 

\item Discretized Laplacian, gradient, and divergence operators defined directly on the lifted space. 

\item Experimental evaluation of segmentation, geodesic approximation, and curvature estimation on non-manifold point clouds and validation that the blow-up yields correct curvature estimates on parametric immersions where existing methods fail near singularities.
\end{itemize}

\section{Related Work}
\paragraph*{Nash blow-ups in algebraic geometry.}
The Nash blow-up pairs each point of a surface with its tangent plane and takes the closure, filling in limiting tangent planes over the singularities. Introduced by Semple~\cite{Semple1954} and Nash~\cite{nash1996arc} and studied by Nobile~\cite{Nobile1975}, Gonzalez-Sprinberg~\cite{GonzalezSprinberg1977}, and Hironaka~\cite{Hironaka1983}, the construction was originally motivated by the question of whether iterated Nash blow-ups can resolve all singularities in algebraic varieties. For curves in characteristic zero, the answer is affirmative~\cite{Nobile1975,Rebassoo1977}; in dimension four and higher, counterexamples exist~\cite{CastilloEtAl2026}.
This distinction is important for our setting since we do not claim that iterated tangent blow-ups resolve arbitrary non-smooth geometry. Instead, our work imports the tangent-separation idea of the Nash blow-up into geometry processing, where the goal is to define operators on point clouds.

\paragraph*{Varifolds.}
A $d$-varifold is a (Radon) measure on $\mathbb{R}^n \times
\mathrm{Gr}(d,n)$, encoding a joint distribution of mass and unoriented tangent planes. The framework naturally accommodates singularities and tangent multiplicity, and has been used for surface approximation~\cite{BuetLeonardiMasnou2017}, curvature estimation~\cite{BuetLeonardiMasnou2022}, and mean curvature flow of point clouds~\cite{BuetRumpf2022}. Our lifted representation uses the same product space, but varifold methods define geometric quantities through integration against test forms, while our construction produces a Euclidean point cloud.

\paragraph*{Feature-sensitive embeddings.}
Several geometry processing methods incorporate normals into distance computations to preserve geometric features. The isophotic metric of Pottmann et al.~\cite{PottmannEtAl2004} weights curve lengths on a triangle mesh by the variation of surface normals. Lai et al.~\cite{LaiEtAl2006} embed mesh vertices and normals into $\mathbb{R}^6$ for feature classification and editing. Related constructions lift meshes into position-normal product spaces so that isotropic operations in the lifted space yield anisotropic or feature-aligned results in the original domain, with applications to remeshing~\cite{CanasGortler2006, LevyBonneel2013}. These methods share with our work the idea of augmenting spatial coordinates with tangent or normal data, but they differ in scope and assumptions. They operate on manifold triangle meshes and assign a single oriented normal per vertex. On a non-manifold mesh, one vertex can have multiple normals, and the augmented embedding does not separate them.

\paragraph*{Modified operators on manifolds.}
Bilateral filters suppress averaging across discontinuities by weighting distances using both spatial proximity and feature similarity~\cite{TomasiManduchi1998,Jones2003, Fleishman2003}. Anisotropic diffusion operators~\cite{ AndreuxEtAl2014,BoscainiEtAl2016} control the direction of information flow via position-dependent diffusion tensors on the tangent plane. These methods modify operators to incorporate geometric information, but they are task-specific and do not define a global geometric representation.

\paragraph*{Singularity detection and resolution.}
In mesh processing, non-manifold singularities are typically resolved as a preprocessing step where edges and vertices are classified combinatorially, and the mesh is split or repaired before downstream algorithms~\cite{ZhouEtAl2016, ValqueLazard2025, charton2021mesh, attene2013polygon}. Sharp and Crane~\cite{sharp2020laplacian} define an intrinsic Delaunay triangulation that handles non-manifold edges directly and yields a well-defined Laplacian without mesh repair, but their construction operates on mesh connectivity and does not use tangent plane information. For point clouds, recent work detects singularities via local topological or statistical tests~\cite{von2023topological, lim2025hades}, but these methods detect singularities without defining differential operators.

\paragraph*{Curvature estimation.}
Point cloud curvature estimation typically proceeds by fitting a local model to the $k$-nearest neighbors. For example, jet fitting~\cite{CazalsPouget2005} uses truncated Taylor expansions, while corrected curvature measures (CNC)~\cite{LachaudRomonThibert2022, LachaudCoeurjolly2023} integrate Lipschitz-Killing forms over randomly sampled triangles. Classical normal-difference approaches~\cite{Rusinkiewicz2004, cao2021efficient} estimate curvature from how normals vary across mesh edges or over local neighborhoods of point clouds. These methods build neighborhoods in the ambient space $\mathbb{R}^3$ and assume that points are sampled from a smooth manifold.

\paragraph*{Stratification learning.}
A separate line of work recovers the topological structure of stratified spaces from data, using tools such as persistent intersection homology~\cite{Bendich2011}, local cohomology and homology transfer~\cite{Bendich2012, Nanda2020}, geometric anomaly detection~\cite{Stolz2020}, subspace clustering~\cite{Vidal2005}, and mixed-density estimation~\cite{Haro2006}. These methods identify non-manifold structure but do not resolve singularities or discretize differential operators, which is the focus of our work.
Most closely related to our work, Tinarrage~\cite{tinarrage2023recovering} estimates the tangent bundle of a sampled immersed manifold and applies persistent homology to the resulting lifted measure. This shares our principle of separating self-intersections, while aiming at topological inference from a single lift rather than the differential operators and iterated higher-order lifts developed here.

\paragraph*{Point cloud reconstruction and processing.}
Neural implicit methods reconstruct surfaces from noisy point samples~\cite{erler2020points2surf,ma2021neural}. Other recent reconstruction methods include Symmetrized Poisson Reconstruction~\cite{kohlbrenner2025symmetrized}, which encodes each sample by the outer product of its normal, and EdgeMovingNet~\cite{yang2025edgemovingnet}, which jointly predicts edge features and normals to preserve sharp structures. Recent work has also produced increasingly accurate estimators for local differential quantities; we refer the reader to the survey of Arnal--Anger et al.~\cite{arnal2026survey}. Similarly, learning-based estimators such as PCPNet~\cite{guerrero2018pcpnet} and DeepFit~\cite{ben2020deepfit} extract normals and higher-order quantities such as curvature from point clouds. However, nearly all the methods above rely on local neighborhoods in the ambient space or assume a manifold structure. Our construction is meant to enable existing algorithms, like those above, to apply to non-manifold data without additional priors or per-task training.

\paragraph*{Normal estimation.}
Our discrete blow-up assumes that each sample carries at least one tangent plane, equivalently an unoriented normal direction for surfaces in $\mathbb{R}^3$. Robustly estimating normals from raw point clouds is itself a substantial problem. Classical approaches include local PCA~\cite{kambhatla1997dimension}, while recent learned estimators improve robustness to noise, outliers, density variation, and orientation ambiguity~\cite{li2022hsurf, li2025high}. These methods extract tangent information and are therefore complementary to our work, with better estimation improving the lifted representation directly. We rely on this in our evaluation: Figure~\ref{fig:rooms} uses HSurf-Net~\cite{li2022hsurf} normals on real indoor scans, while Figures~\ref{fig:tangent-noise} and~\ref{fig:noise-study} use local PCA tangents under controlled noise.

\section{Mathematical Preliminaries}
At a singular point, multiple tangent planes coexist at the same position.
The tangent blow-up resolves this ambiguity by recording which tangent plane is associated to each point. Rather than representing a sample by its position $x \in \mathbb{R}^n$ alone, we represent it by pairing each position with its tangent plane. Singular points that coincide in space but have different tangents become multiple, distinct points in the lifted representation. For example, at the transverse self-intersection of a figure-eight, the singular point $x_0$ is replaced by two lifted points,$(x_0, T_1 {\nearrow})$ and $(x_0, T_2 {\searrow})$ (see Figure~\ref{fig:blowup-intuition}). 

\begin{figure}[htb]
  \centering
  \includegraphics[width=.8\linewidth]{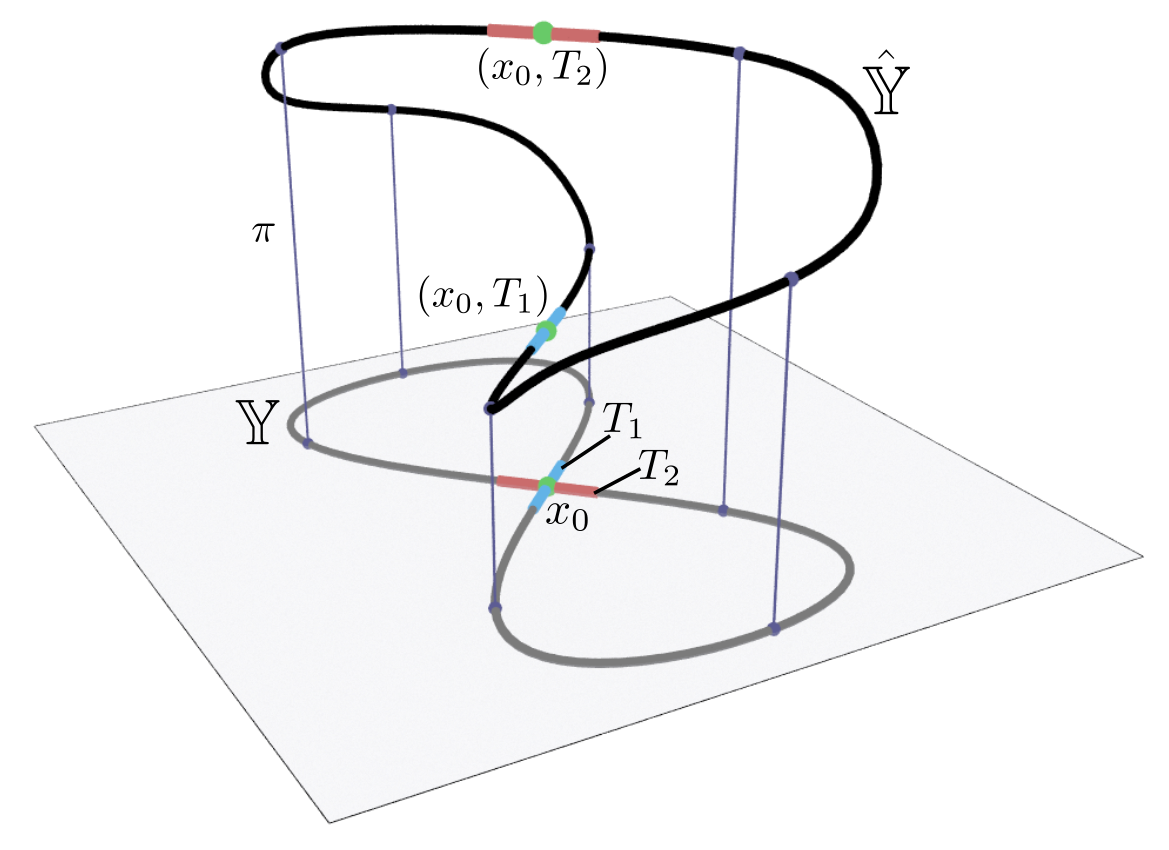}
  \caption{\label{fig:blowup-intuition} A schematic cartoon illustrating the tangent blow-up of a figure-eight curve. In the base space, the curve crosses itself at $x_0$. After lifting to the product space $\mathbb{R}^n \times \mathrm{Gr}(d,n)$, the singular point is replaced by two lifted points $(x_0, T_1 {\nearrow})$ and $(x_0, T_2 {\searrow})$. The projection $\pi$ (defined in \S\ref{sec:background}) is an isomorphism everywhere away from the set of singular points.}
\end{figure}


\subsection{Continuous Construction}\label{sec:continuous}
We begin by making our continuous construction precise. We first define the class of singular spaces we assume our point clouds are sampled from, then introduce the Grassmannian of tangent planes together with a suitable computational representation for its elements via \emph{projector} matrices, and finally define the Nash blow-up.
Because every object defined here admits a concrete matrix representation, the blow-up procedure ultimately reduces to concatenating position coordinates with the vectorized projector matrices.
This allows for simplified linear algebra and avoids the need for orienting global normal fields.

\paragraph*{Stratified spaces.} Recall that a smooth $d$-dimensional manifold is a space that locally resembles $\mathbb{R}^d$ at every point~\cite{lee2003smooth}. Stratified spaces generalize this definition to address singular spaces. Informally, a stratified space $\mathbb{Y}$ is a topological space that can be decomposed into manifold pieces meeting along lower-dimensional singular sets (see Figure~\ref{fig:manifold-vs-stratified}); we refer the reader to Goresky and MacPherson~\cite{goresky1988stratified} or Friedman~\cite{friedman2020singular} for formal treatments.
This is a natural setting for geometry processing since many ``real-world singularities,'' such as intersecting components of CAD models, produce this structure. We refer to the union of $d$-dimensional manifold pieces as the \emph{regular part} $\mathbb{Y}_{\mathrm{reg}}$
and assume each manifold piece of $\mathbb Y_{\mathrm{reg}}$ is smooth.

\begin{figure}[htb]
  \centering
  \includegraphics[width=\linewidth]{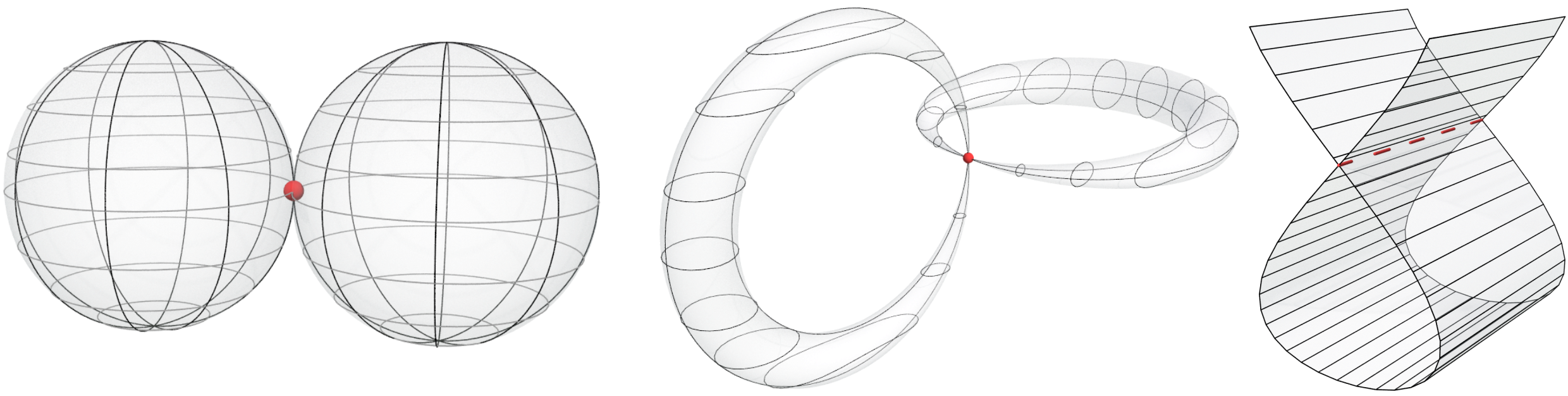}
  \caption{\label{fig:manifold-vs-stratified} Examples of stratified spaces.}
\end{figure}

\paragraph*{The Grassmannian and its projector representation.}
To represent and compute over (point, tangent) pairs, we need a space that parametrizes all possible tangent planes.
This is the \emph{Grassmannian} $\mathrm{Gr}(d,n)$, the compact smooth manifold of $d$-dimensional linear subspaces of $\mathbb{R}^n$~\cite{lee2003smooth}.
The Grassmannian admits multiple concrete computational representations~\cite{bendokat2024grassmann}.
We opt to represent each element of $\mathrm{Gr}(d,n)$ by its orthogonal projector $P = UU^\top \in \mathbb{R}^{n \times n}$, where $U \in \mathbb{R}^{n \times d}$ is any orthonormal basis of the subspace. The projector is independent of the choice of basis, and allows one to map $\mathrm{Gr}(d,n)$ injectively into the space of symmetric matrices. We also define the vectorization operator $\mathrm{vec} : \mathbb{R}^{n \times n} \to \mathbb{R}^{n^2}$, which stacks the columns of a matrix into a single vector. 
With this, distances between subspaces become Frobenius distances between projector matrices, which after vectorization become ordinary Euclidean distances~\cite{golub2013matrix}.

\paragraph*{The Nash blow-up.} With the Grassmannian and its projector representation, we can now define the Nash blow-up. The idea is to pair each non-singular point with its tangent plane and take the closure to fill in the singular points.

Let $\mathbb{Y}$ be a stratified space.
On the regular part, each point $x \in \mathbb{Y}_{\mathrm{reg}}$ has a unique tangent $d$-plane $T_x \mathbb{Y} \in \mathrm{Gr}(d,n)$. The \emph{generalized Gauss map} $$ \nu : \mathbb{Y}_{\mathrm{reg}} \longrightarrow \mathbb{R}^n \times \mathrm{Gr}(d,n), \qquad \nu(x) = (x, T_x \mathbb{Y}), $$ sends each regular point to the pair of its position and its tangent plane, embedding $\mathbb{Y}_{\mathrm{reg}}$ into the product of the ambient space with the Grassmannian (see Figure~\ref{fig:generalized-gauss}). The \emph{Nash blow-up} of $\mathbb{Y}$ is the \emph{closure} of this image, which introduces limit points of convergent sequences on $\mathbb Y$:
$$
\widehat{\mathbb{Y}} \;=\; \overline{\bigl\{\,(x,\, T_x \mathbb{Y}) \;\big|\; x \in \mathbb{Y}_{\mathrm{reg}}\,\bigr\}} \;\subset\; \mathbb{R}^n \times \mathrm{Gr}(d,n).
$$
Over the regular part, $\widehat{\mathbb{Y}}$ is the graph of the tangent-plane assignment.
Over the set of singular 
points, however, something more interesting happens when we take the closure: by definition, a point $(x_0, P) \in \widehat{\mathbb{Y}}$ with $x_0 \in \Sigma$ exists if and only if $P$ is the limit of tangent planes along some sequence of regular points converging to $x_0$. In other words, the closure fills in \emph{limiting} tangent planes at each singularity. Returning to the figure-eight example (Figure \ref{fig:blowup-intuition}), the crossing point $x_0$ acquires two lifted copies $(x_0, T_1)$ and $(x_0, T_2)$.

The \emph{projection} $\pi : \widehat{\mathbb{Y}} \to \mathbb{Y}$, $\pi(x,P) = x$, maps the blown-up space back to the original by ``forgetting'' the tangent plane. This map is continuous, surjective, and restricts to a homeomorphism over $\mathbb{Y}_{\mathrm{reg}}$, so the blow-up modifies $\mathbb{Y}$ only over the singular points.
The preimage $\pi^{-1}(x_0)$ of a singular point is its \emph{exceptional fiber}. Each point in the exceptional fiber corresponds to a distinct limiting tangent plane; we call these branches \emph{sheets}. For example, the figure-eight has two sheets at its crossing point, each carrying its own tangent.

\begin{figure}[htb]
  \centering
  \includegraphics[width=\linewidth]{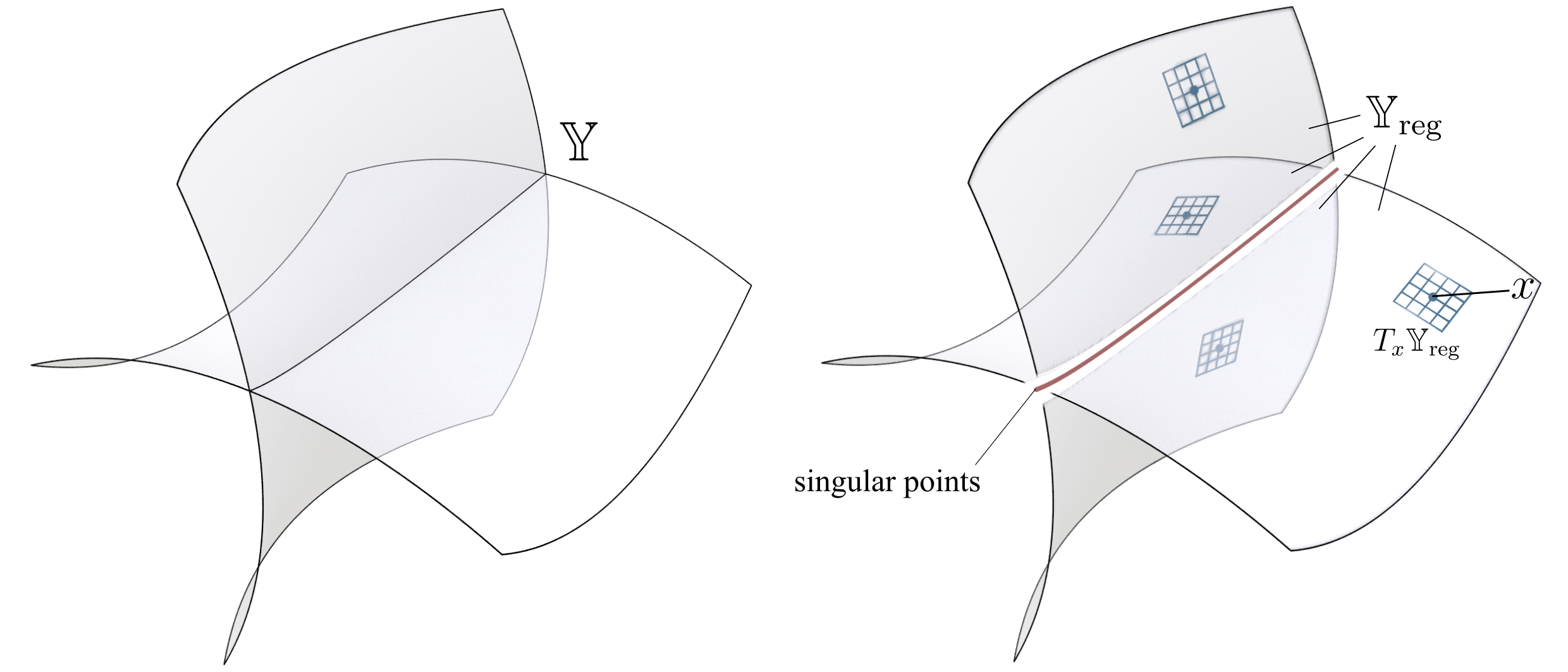}
  \caption{\label{fig:generalized-gauss} Left: A stratified space consisting of four smooth manifold sheets intersecting along singular points. Right: The space $\mathbb{Y}$ is decomposed into its regular part $\mathbb{Y}_{\mathrm{reg}}$ and singular set of points (shown in red). Each smooth sheet is then lifted via the generalized Gauss map, associating to each regular point its tangent space.}
\end{figure}

\subsection{Iterated Blow-Ups}\label{sec:lifted-tangents}
The blow-up, as presented thus far, only separates transversely intersecting points with distinct tangent planes.
This construction can be extended beyond transverse intersections to resolve higher-order contact phenomena, including osculating structures, under suitable conditions.
The extension is achieved by iterating the blow-up, lifting successively from the previously lifted space.
While transverse intersections are common in, e.g., messy geometry, iterated blow-ups are practically useful for separating out point clouds of multiple touching or proximate objects (see Figure~\ref{fig:teaser}).

\paragraph*{Projector derivative.}\label{sec:projector-derivative}
Iterating upon the first blow-up necessitates formally defining tangents at each lifted point $(x, P)$.
Because our blow-up is a product space, we know that the tangent space at $(x, P)$ decomposes into a direct sum of the tangent space at $x \in \mathbb R^n$, the so-called \emph{horizontal} component, and the tangent space at $P \in \mathrm{Gr}(d,n)$, the \emph{vertical} component:
$$
T_{(x, P)}\left(\mathbb{R}^n \times \operatorname{Gr}(d, n)\right)=\underbrace{T_x \mathbb{R}^n}_{\text {horizontal}} \oplus \underbrace{T_P \operatorname{Gr}(d, n)}_{\text {vertical }}.
$$

The projector field $x \mapsto P$ assigns to each $x \in \mathbb{Y}_{\mathrm{reg}} \subset \mathbb{R}^n$ the orthogonal projector onto the tangent plane $T_x \mathbb{Y}$.
For a curve $c(t)$ on $\mathbb{Y}_{\mathrm{reg}}$ with $c(0)=x$ and $c'(0)=v \in T_x\mathbb{Y}$, the lifted curve $t \mapsto (c(t), P(c(t)))$ has tangent
\[
\frac{d}{dt}(c(t), P(c(t)))\bigg|_{t=0}
= \big(v, \nabla_v P\big),
\]
where $\mathcal{P}_v \coloneqq \nabla_v P \in T_P \mathrm{Gr}(d,n)$ denotes the vertical component of the lifted tangent space.
Differentiating the idempotency condition
$P^2 = P$ along $v$ gives $\mathcal{P}_v P + P\mathcal{P}_v = \mathcal{P}_v$, from which it follows that
$$
P\,\mathcal{P}_v\,P = 0
\qquad\text{and}\qquad
(I - P)\,\mathcal{P}_v\,(I - P) = 0.
$$
This means $\mathcal{P}_v$ has no tangent-to-tangent or normal-to-normal component, so it maps tangent vectors into the normal space and normal vectors into the tangent space. 
Indeed, any $n \times n$ matrix satisfying these two conditions is a tangent vector at $P \in \mathrm{Gr}(d,n)$~\cite{boumal2023introduction, bendokat2024grassmann}.


\paragraph*{Curvature.} 
As a point moves along the surface, its tangent plane rotates, and the rate of this rotation encodes curvature. 
The projector derivative is therefore naturally related to the curvature of $\mathbb{Y}_{\mathrm{reg}}$.
Specifically, it represents the normal component of the differential change in tangent vector $w$ in direction $v$, i.e., it is the same as the vector-valued second fundamental form, $\mathrm{II}(v, w) = \mathcal{P}_v(w)$, which is a generalization of the second fundamental form to arbitrary codimension~\cite{robbin2022introduction} (see Figure~\ref{fig:projector-derivative}).
The \emph{scalar} valued second fundamental form, more commonly used in geometry processing, picks a unit normal $\ell$ and returns a number $\mathrm{II}(v, w)=\left\langle \mathcal{P}_v(w), \ell\right\rangle \in \mathbb{R}$ rather than a vector in the normal space.

The projector derivative $\mathcal{P}_v$ can be represented as an $n \times n$ matrix, but its subspace structure (tangent-to-normal and normal-to-tangent only) means it is fully determined by the $(n{-}d) \times d$ matrix
$$
  B_x(v) \;=\; N^\top \mathcal{P}_v\, U
    \;\in\; \mathbb{R}^{(n-d) \times d},
$$
where $U \in \mathbb{R}^{n \times d}$ and $N \in \mathbb{R}^{n \times (n-d)}$ are orthonormal bases for the tangent and normal spaces at
$x$~\cite{bendokat2024grassmann}.
Each entry $B_x(v)_{\ell k}$ records how much the $k$-th tangent basis vector rotates into the $\ell$-th normal direction when the basepoint moves along $v$. The (vector valued) second fundamental form is recovered by mapping these coordinates back to ambient space: $\mathrm{II}(v, w) = N\, B_x(v)\,(U^\top w)$.

For hypersurfaces ($n - d = 1$), the normal space is one-dimensional and $B_x(v)$ reduces to a row vector for any tangent direction $v$. Stacking the rows $B_x(e_1), \ldots, B_x(e_d)$ retrieves the shape operator, whose eigenvalues are the principal curvatures $\kappa_1, \ldots, \kappa_d$.

\paragraph*{Computation.}

Computing $\mathcal{P}_v \in \mathbb{R}^{n \times n}$ directly requires differentiating the projector field. In Section~\ref{sec:iterated-blowup} we show how to estimate the compact representation $B_x$ from a point
cloud by linear regression, after which the full projector derivative can be reconstructed as
$$
  \mathcal{P}_v
    = N\,B_x(v)\,U^\top + U\,B_x(v)^\top N^\top.
$$
While $B_x(v)$ depends on the choice of tangent and normal bases $U$ and $N$, the reconstructed $\mathcal{P}_v$ does not~\cite{bendokat2024grassmann}.

\begin{figure}[htb]
  \centering
  \includegraphics[width=.8\linewidth]{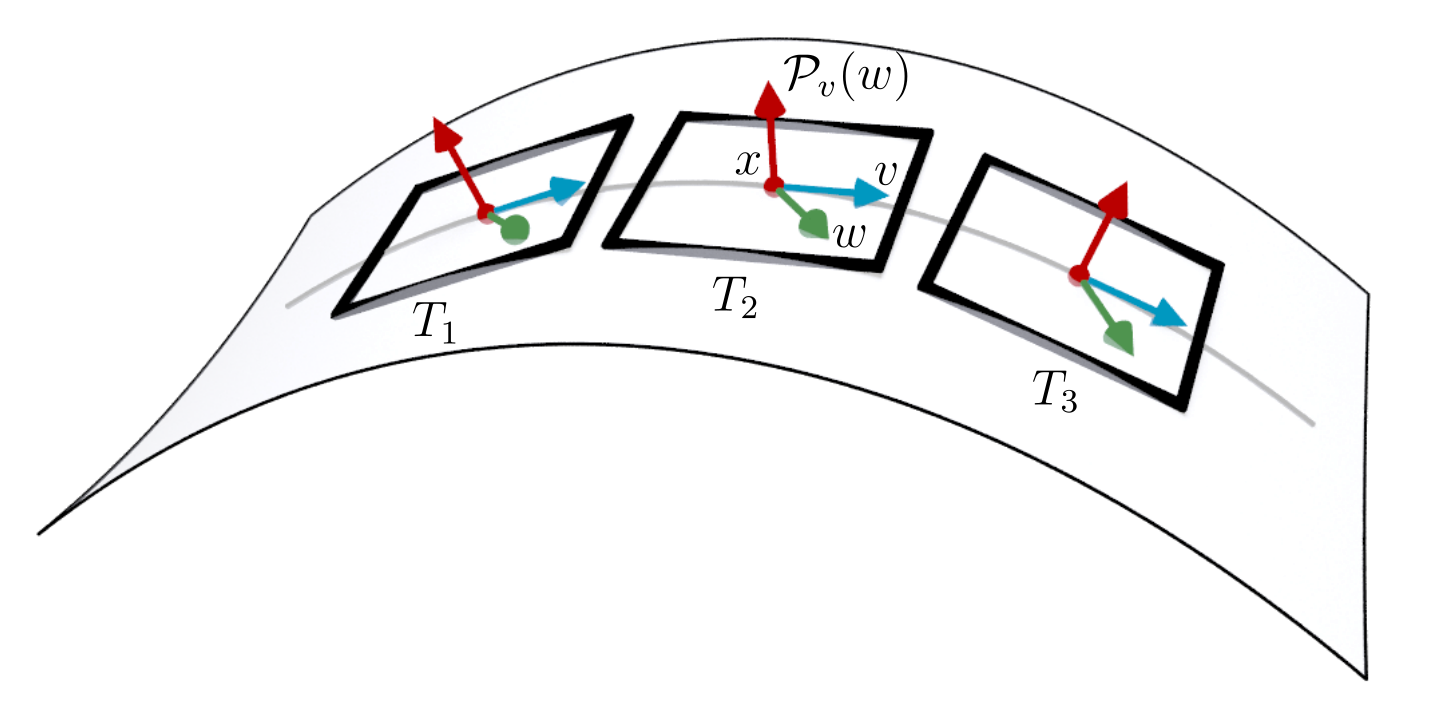}
  \caption{\label{fig:projector-derivative}The projector derivative on a smooth surface. As the basepoint moves along the curve with tangent $v$ (blue arrow), the tangent plane rotates from $T_1$ to $T_2$ to $T_3$. At $x_1$, the projector derivative $\mathcal{P}_v(w)$ (red arrow) measures the component of the tangent vector $w$ (green arrow) that leaves the tangent plane, which is the vector valued second fundamental form $\mathrm{II}(v, w)$.
  }
\end{figure}\label{sec:background}

\section{The Discrete Tangent Blow-Up}
We now present our main contribution: 
a discrete representation that makes the blow-up algorithmic and applicable to point cloud data. We begin by defining a metric on the product space $\mathbb{R}^n \times \mathrm{Gr}(d,n)$ and its Euclidean embedding, then describe the discrete lifting algorithm that takes a set of points and tangents as input and produces a point cloud  with resolved singularities in the lifted space as output.

\paragraph*{Lifted metric.}\label{sec:lifted-metric}
All discretized operators on the lifted point cloud, such as the Laplacian (\S\ref{sec:lifted-laplacian}), gradient (\S\ref{sec:lifted-grad-div}), and the second fundamental form regression (\S\ref{sec:iterated-blowup}), require a notion of neighborhood and proximity. To perform nearest-neighbor search and kernel evaluation in the product space $\mathcal{L} = \mathbb{R}^n \times \mathrm{Gr}(d,n)$, we define a metric that combines spatial and tangential components. The \emph{chordal distance} on the Grassmannian,
\[
  d_{\mathrm{chord}}(U,V) \;=\; \tfrac{1}{\sqrt{2}}\,\|P_U - P_V\|_F \;=\; \Bigl(\textstyle\sum_{i=1}^d \sin^2\theta_i\Bigr)^{1/2},
\]
where $\theta_1, \ldots, \theta_d$ are the principal angles between $U, V \in \mathrm{Gr}(d,n)$, is computationally efficient, and its square is smooth everywhere~\cite{conway1996packing}.
Furthermore, it approximates the geodesic distance on $\mathrm{Gr}(d,n)$ with cubic error as distances approach zero, $d_{\mathrm{chord}} = d_{\mathrm{geo}} + O(d_{\mathrm{geo}}^3)$, as follows directly from the Taylor expansion of the principal-angle characterizations of both metrics~\cite{edelman1998geometry, bendokat2024grassmann}.

We define the \emph{lifted metric} on $\mathbb{R}^n \times \mathrm{Gr}(d,n)$ as the weighted product
\begin{equation}\label{eq:dist}
  d_{\mathcal{L}}^2\bigl((p,U),\,(q,V)\bigr)
  \;=\; \|p - q\|^2 \;+\; \tfrac{\alpha}{2}\,\|P_U - P_V\|_F^2,
\end{equation}
where $\alpha > 0$ is a parameter balancing the spatial and angular parts.
Since the lifted metric measures both spatial displacement and tangent-plane variation, $d_{\mathcal L}$ changes distances by stretching regions where curvature causes tangent planes to vary rapidly. Geodesics and differential quantities computed after lifting should therefore be interpreted as quantities of the lifted, curvature-sensitive geometry, even for a smooth embedded manifold.

\paragraph*{Isometric Euclidean embedding.} 
A useful computational property of the chordal product metric is that it admits an isometric embedding into Euclidean space. Define the embedding map
\[
  \Phi : \mathbb{R}^n \times \mathrm{Gr}(d,n) \to \mathbb{R}^{n + n^2}, \qquad
  \Phi(p,\,U) = \bigl(\,p,\;\; \sqrt{\smash[b]{\alpha/2}}\;\mathrm{vec}(P_U)\,\bigr),
\]
where $\mathrm{vec}(\cdot)$ denotes vectorization. Then $\|\Phi(p,U) - \Phi(q,V)\|^2 = d_{\mathcal{L}}^2((p,U),(q,V))$. Thus, Euclidean distance in the embedding space exactly reproduces the lifted metric.
Consequently, computing neighborhoods and kernels in the lifted metric reduces to applying ordinary Euclidean algorithms to the coordinates $\Phi(p,U)$.

\subsection{Algorithm}\label{sec:discrete-algorithm}

\paragraph*{Input.}
The input to our discrete blow-up algorithm is a point cloud equipped with tangent planes: positions $\{x_i\}_{i=1}^N \subset \mathbb{R}^n$, orthonormal tangent frames $\{U_i\}_{i=1}^N$ 
with $U_i \in \mathbb{R}^{n \times d}$ spanning the tangent $d$-plane at point $x_i$ with index $i$, and a weight parameter $\alpha > 0$ used in our metric in \eqref{eq:dist}. For surfaces in $\mathbb{R}^3$, a tangent plane is equivalent to a normal vector, though no globally consistent orientation is required. For curves, a unit tangent vector suffices. We assume that tangent planes are provided; the practical problem of tangent estimation near singularities is discussed in \S\ref{sec:discussion}.

\paragraph*{The lift.}
For each sample, we compute the tangent-plane projector $P_i = U_iU_i^\top \in \mathbb{R}^{n \times n}$, the orthogonal projection onto $\mathrm{col}(U_i)$. The projector depends only on the subspace spanned by $U_i$, not on the choice of basis. Consequently any quantity derived from $P$, such as distances between projectors and derivatives of the projector field, is also coordinate-free.
We then apply the isometric embedding:
$$
  \Phi_i \;=\; \Bigl(\,x_i,\;\; \sqrt{\alpha/2}\;\,\mathrm{vec}(P_i)\,\Bigr) \;\in\; \mathbb{R}^{n + n^2}.
$$
The output is a Euclidean point cloud $\{\Phi_1, \ldots, \Phi_N\} \subset \mathbb{R}^{n+n^2}$ on which Euclidean distances reproduce the lifted metric by construction:
$$
\|\Phi_i - \Phi_j\|^2 \;=\; \|x_i - x_j\|^2 + \tfrac{\alpha}{2}\,\|P_i - P_j\|_F^2 \;=\; d_{\mathcal{L}}^2\bigl((x_i, P_i)\,, (x_j, P_j)\bigr).
$$
The coordinate transformation is applied uniformly to every sample point; no knowledge of singular points or any global topology is required.

To understand how the lifting separates intersecting sheets, consider two points $x_i$ and $x_j$ near a transverse self-intersection. In the ambient space $\mathbb{R}^n$, these points may be arbitrarily close. In the lifted space, however, the projector term $\tfrac{\alpha}{2}\|P_i - P_j\|_F^2$ contributes a penalty proportional to the angular gap between their tangent planes. Larger $\alpha$ widens the gap at the cost of stretching distances along curved regions where the tangent plane varies rapidly. Figure~\ref{fig:fig8-coordinates} illustrates this effect on a self-intersecting figure-eight curve, where the projector entries vary smoothly along each branch but differ sharply at the crossing; we offer a formal treatment in \S\ref{sec:theoretical-analysis}.

\begin{figure}[htb]
  \centering
  \includegraphics[width=.8\linewidth]{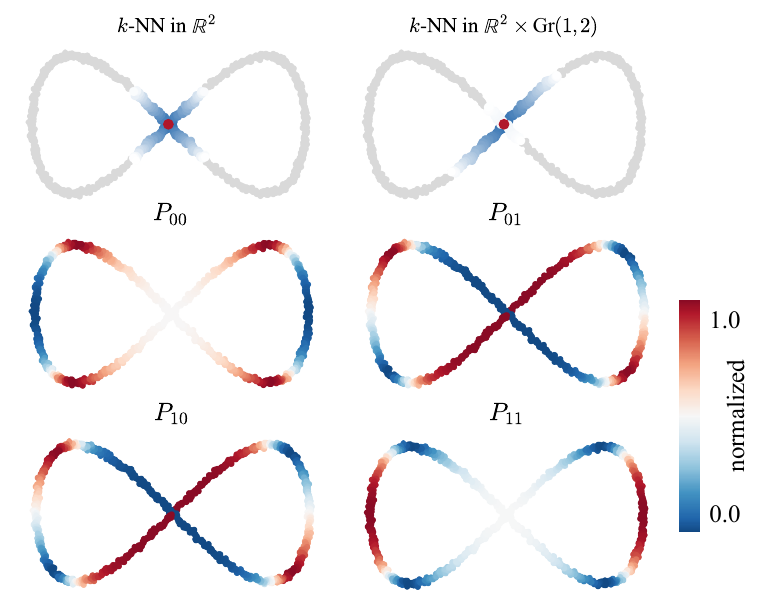}
  \caption{\label{fig:fig8-coordinates}The discrete tangent blow-up of a self-intersecting figure-eight curve. Top row: $k$-nearest neighbors of a query point (red), computed in the ambient space $\mathbb{R}^2$ (left) and in the lifted embedding $\mathbb{R}^2 \times \mathrm{Gr}(1,2)$ (right). Bottom rows: entries of the tangent projector $P \in \mathrm{Gr}(1,2)$, shown component-wise over the curve. Points at the crossing share spatial position but have distinct tangents, yielding different values of $P$.}
\end{figure}

\begin{remark}[The exceptional fiber in the discrete setting.] At a singular point $x_0$, the continuous blow-up $\widehat{\mathbb{Y}}$ has a well-defined exceptional fiber (defined in \S\ref{sec:background}). In the discrete setting, we do not explicitly compute a closure, but the fiber structure is inherited by the samples. A point $x_i$ sampled near a singularity $x_0$ on a sheet $M$ carries a tangent estimate $P_i \approx T_{x_0} M$, so its lifted image $\Phi_i$ lies near the fiber point $(x_0, T_{x_0} M)$. As the sampling density increases, the lifted point cloud approximates this structure without any explicit closure.
\end{remark}

\subsection{The Discrete Iterated Blow-Up}\label{sec:iterated-blowup}

The first lift separates sheets with distinct tangent planes. However, multiple sheets can share the same tangent at a point and still be geometrically distinct. For example, a line and a parabola meeting tangentially at the origin have the same position and tangent, but different curvature (see Figure~\ref{fig:second-blowup}). 
To resolve such tangential intersections, we need to capture how the tangent plane varies from point to point.

The discrete blow-up can be applied iteratively: the output of the first lift is a Euclidean point cloud $\{\Phi_i\}$, and the tangent projector at each lifted point is itself a point on a new Grassmannian.
Hence, we can define a new product space to which the entire blow-up machinery applies again. 
To perform this second lift, we need the tangent planes of the lifted product space $\mathbb{R}^n \times \mathrm{Gr}(d,n)$, which requires estimating how $P_i$ varies across the point cloud (recall \S\ref{sec:lifted-tangents}). This variation is encoded by the vector valued second fundamental form.

\paragraph*{Discrete estimation of the second fundamental form.}
Although the lifted point cloud lives in $\mathbb{R}^{n+n^2}$, the second fundamental form can be estimated from the original data $(x_i, U_i)$ alone, without performing PCA or tangent estimation in the high-dimensional lifted space. This is because the vertical component of the product tangent vector is determined entirely by the projector derivative, which depends only on $n$-dimensional positions and $d$-dimensional tangent frames.

Broadly speaking, we approximate the vector valued second fundamental form $\mathcal{P}_v$ at each sample $x_i$ by finite-differencing the projector $P_i$ against neighboring projectors, and recovering a discrete estimate of $B_i \approx B_{x_i}$ by linear regression.
Thus, the procedure to estimate $B_i$ is as follows. 
For each point $(x_i, U_i)$, fix an orthonormal basis $N_i \in \mathbb{R}^{n \times (n-d)}$ for $\mathrm{col}(U_i)^\perp$. Throughout, $a, b = 1, \ldots, d$ index tangent directions with $e_a$ the standard basis of $\mathbb{R}^d$, and $\ell = 1, \ldots, n-d$ indexes normal directions. Then:
\begin{enumerate}
  \item Find $k$-nearest neighbors of $\Phi_i$ in the lifted metric.
  \item Compute intrinsic displacements $t_{ij} = U_i^\top(x_j - x_i) \in \mathbb{R}^d$: the spatial offset projected into the tangent plane at $x_i$.
  \item Compute Grassmannian tangent coordinates $C_{ij} = N_i^\top(P_j - P_i)\, U_i \in \mathbb{R}^{(n-d) \times d}$: the projector difference expressed in the tangent and normal bases at $x_i$. For nearby points on the same smooth sheet, $C_{ij} \approx B_{x_i}(t_{ij}) + O(\|x_j - x_i\|^2)$.
  \item Solve the constrained least-squares problem:
  \begin{align*}
    B_i = \operatorname*{argmin}_{B} &\sum_{j \in \mathcal{N}(i)}
    \|C_{ij} - B\, t_{ij}\|_F^2 \\
    \text{subject to} \quad
      &(B(e_a))_{\ell b} = (B(e_b))_{\ell a}
      \quad \text{for each normal direction } \ell.
  \end{align*}
\end{enumerate}
The symmetry constraint reflects the symmetry of the vector valued second fundamental form on smooth manifolds~\cite{lee2003smooth}. 
The resulting $B_i \approx B_{x_i}$ is the discrete analogue of the matrix defined in \S\ref{sec:projector-derivative}, from which the full projector derivative is reconstructed as $\mathcal{P}_v = N_i\,B_i(v)\,U_i^\top + U_i\,B_i(v)^\top N_i^\top$ and all curvature invariants (\S\ref{sec:curvature}) are extracted.
In the hypersurface case ($n - d = 1$), $N_i$ reduces to a single unit normal, $B_i(e_a) \in \mathbb{R}^{1 \times d}$ is a row vector for each tangent direction, and the symmetry constraint makes the $d \times d$ matrix with rows $B_i(e_1), \ldots, B_i(e_d)$ symmetric. Hence, we recover the shape operator.

\paragraph*{Constructing the second blow-up.}\label{sec:second-lift}
Having estimated each $B_i$, we can construct the tangent space of the lifted manifold at each point. For each tangent basis direction $e_a \in \mathbb{R}^d$, the tangent vector to the first lift at $(x_i, P_i)$ has a horizontal component (the spatial velocity $U_i\, e_a$) and a vertical component (the corresponding Grassmannian velocity, encoding how the projector changes when the basepoint moves along $e_a$).

The vertical component is reconstructed from its coordinate representation $B_i(e_a)$ via the tangent space characterization of the Grassmannian~\cite{bendokat2024grassmann}: any tangent vector to $\mathrm{Gr}(d,n)$ at $P_i$ takes the form $N_i B U_i^\top + U_i B^\top N_i^\top$ for some $B \in \mathbb{R}^{(n-d) \times d}$.
\[
  g_a = \Bigl(\, U_i\, e_a,\;\; \sqrt{\alpha/2}\;\, \mathrm{vec}\bigl(N_i\, B_i(e_a)\, U_i^\top + U_i\, B_i(e_a)^\top\, N_i^\top\bigr)\,\Bigr) \;\in\; \mathbb{R}^{n + n^2}.
\]
The vectors $g_1, \ldots, g_d$ span the tangent space of the lifted manifold at $\Phi_i$, but may not in general be orthonormal. 

To obtain a tangent projector, we stack them into $G_i = [g_1 \;\cdots\; g_d] \in \mathbb{R}^{(n+n^2) \times d}$
and form
\[
  Q_i = G_i (G_i^\top G_i)^{-1} G_i^\top \in \mathrm{Gr}(d,\; n + n^2).
\]
Although the intermediate quantities $B_i$, $g_a$, and $G_i$ depend on the choice of bases $U_i$ and $N_i$, the projector $Q_i$ does not. 
This is the same coordinate-independence that holds for $P_i$ itself, as noted in \S\ref{sec:background}.

The second lift applies the blow-up to its own output. After the first lift, each sample carries a position $\Phi_i^{(1)} \in \mathbb{R}^{D_1}$ (with $D_1 = n + n^2$) and a tangent projector $P_i^{(1)} = Q_i \in \mathrm{Gr}(d, D_1)$, computed from the procedure above. The goal of the second lift is to produce a new tangent projector $P_i^{(2)} \in \mathrm{Gr}(d, D_2)$ that encodes how $P_i^{(1)}$ itself varies across the point cloud. The second-lift embedding and its associated tangent projector are constructed by applying the same algorithm to the pair $(\Phi_i^{(1)}, P_i^{(1)})$ to obtain
$$
\Phi_i^{(2)} = \Bigl(\,\Phi_i^{(1)},\;\; \sqrt{\alpha^{(2)}/2}\;,\mathrm{vec}(P_i^{(1)})\,\Bigr) \;\in\; \mathbb{R}^{D_2}, \qquad D_2 = D_1 + D_1^2
$$
where $\alpha^{(2)} > 0$ is the weight parameter for the second level. The second-lift distance is
$$
d^2_{\mathcal{L}_2} = \|\Phi_i^{(1)} - \Phi_j^{(1)}\|^2 + \tfrac{\alpha^{(2)}}{2}\,\|P_i^{(1)} - P_j^{(1)}\|_F^2
$$
which penalizes differences in position, tangent plane, and curvature simultaneously. The procedure is summarized in Algorithm~\ref{alg:iterated-blowup}.

\begin{figure}
  \centering
  \includegraphics[width=\linewidth]{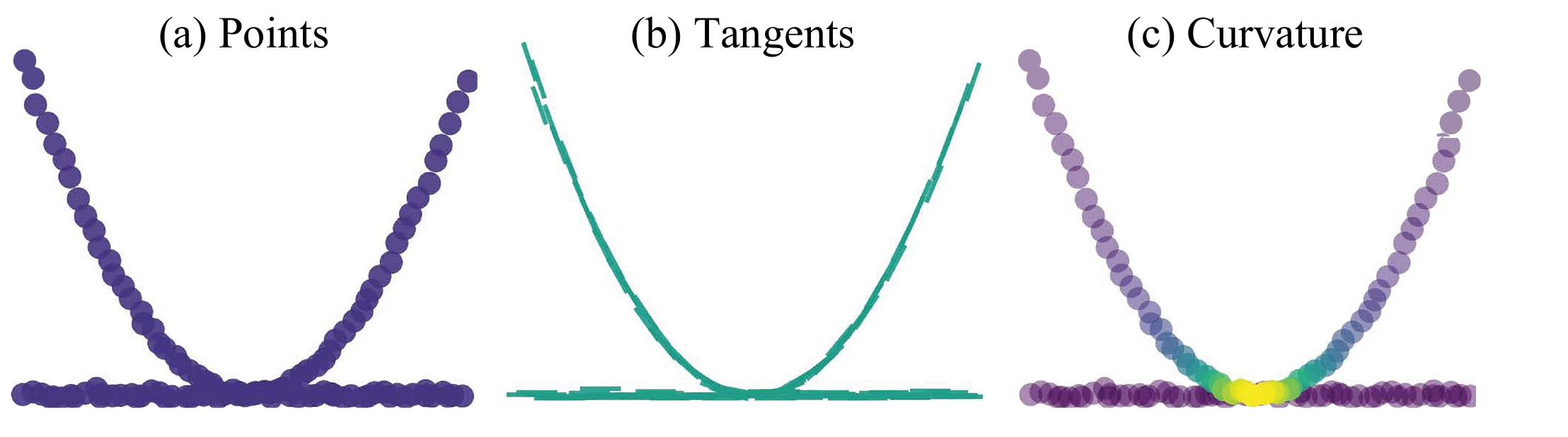}
  \caption{\label{fig:second-blowup}A line and parabola sharing position (left) and tangent (center) at the origin. The Frobenius norm of the vector valued second fundamental form $\|I\!I\|_F$ distinguishes the intersection (right).}
\end{figure}

\begin{algorithm}[htb]
\caption{Discrete iterated blow-up (levels one and two)}
\label{alg:iterated-blowup}
\begin{algorithmic}[1]
\renewcommand{\algorithmicrequire}{\textbf{Input:}}
\renewcommand{\algorithmicensure}{\textbf{Output:}}
\REQUIRE Positions $\{x_i\}_{i=1}^N \subset \mathbb{R}^n$,
  tangent frames $\{U_i\}_{i=1}^N$, weights
  $\alpha^{(1)}, \alpha^{(2)} > 0$, neighbors~$k$
\ENSURE Level-two embedding
  $\{\Phi_i^{(2)}\}_{i=1}^N \subset \mathbb{R}^{D_2}$,
  tangent projectors
  $\{P_i^{(1)}\}_{i=1}^N$,
  $\{P_i^{(2)}\}_{i=1}^N$
\medskip
\STATE \textbf{--- Level-one lift: embed into
  $\mathbb{R}^{D_1}$, $D_1 = n + n^2$ ---}
\FOR{$i = 1, \ldots, N$}
  \STATE $P_i^{(0)} \gets U_i\,U_i^\top$
    \hfill\COMMENT{tangent projector}
  \STATE $\Phi_i^{(1)} \gets
    \bigl(x_i,\;\sqrt{\alpha^{(1)}/2}\;
    \mathrm{vec}(P_i^{(0)})\bigr)$
    \hfill\COMMENT{isometric embedding}
\ENDFOR
\medskip
\STATE \textbf{--- Curvature estimation and tangent projectors ---}
\FOR{$i = 1, \ldots, N$}
  \STATE $N_i \gets \mathrm{orth}(\ker(P_i^{(0)}))$
    \hfill\COMMENT{normal basis}
  \STATE $\mathcal{N}(i) \gets k\text{-NN of } \Phi_i^{(1)}
    \text{ in } \mathbb{R}^{D_1}$
    \hfill\COMMENT{lifted neighbors}
  \FOR{$j \in \mathcal{N}(i)$}
    \STATE $t_{ij} \gets U_i^\top(x_j - x_i)$
      \hfill\COMMENT{tangent-plane displacement}
    \STATE $C_{ij} \gets N_i^\top(P_j^{(0)} - P_i^{(0)})\,U_i$
      \hfill\COMMENT{projector variation}
  \ENDFOR
  \STATE $B_i^{(1)} \gets \operatorname*{argmin}_{B:\,
    B_\ell^\top = B_\ell}
    \sum_j \|C_{ij} - B\,t_{ij}\|_F^2$
    \hfill\COMMENT{fit second fundamental form}
  \FOR{$a = 1, \ldots, d$}
    \STATE $(\mathcal{P}_{e_a})_i \gets N_i\,B_i^{(1)}(e_a)\,U_i^\top
      + U_i\,B_i^{(1)}(e_a)^\top N_i^\top$
      \hfill\COMMENT{Grassmannian velocity}
    \STATE $g_a \gets \bigl(U_i\,e_a,\;
      \sqrt{\alpha^{(1)}/2}\;\mathrm{vec}(\Delta_a)\bigr)$
      \hfill\COMMENT{product tangent vector}
  \ENDFOR
  \STATE $G_i \gets [g_1 \;\cdots\; g_d]$
  \STATE $P_i^{(1)} \gets G_i\,(G_i^\top G_i)^{-1}\,G_i^\top$
    \hfill\COMMENT{tangent projector of the lifted manifold}
\ENDFOR
\medskip
\STATE \textbf{--- Level-two lift: embed into
  $\mathbb{R}^{D_2}$, $D_2 = D_1 + D_1^2$ ---}
\FOR{$i = 1, \ldots, N$}
  \STATE $\Phi_i^{(2)} \gets
    \bigl(\Phi_i^{(1)},\;
    \sqrt{\alpha^{(2)}/2}\;
    \mathrm{vec}(P_i^{(1)})\bigr)$
    \hfill\COMMENT{same structure as level one}
\ENDFOR
\end{algorithmic}
\end{algorithm}

\begin{remark}[Higher-order iteration] In principle, the discrete blow-up construction iterates indefinitely. Each level $\ell$ produces a new Euclidean point cloud with tangent projectors, and each successive level separates points that agree up to one higher order of contact: level zero distinguishes positions, level one distinguishes tangent planes, level two distinguishes curvatures, and so on. In practice, however, each iteration estimates a higher-order information from finite noisy samples, and the accuracy of these estimates likely degrades with order. We conjecture that singularities of contact order $r$ (sheets sharing the same ``$r$-jet'' 
but differing at order $r+1$) are resolved at exactly level $r$. 
In this paper, we restrict attention to levels 0 and 1, which handle transverse and tangential intersections respectively.
Our method is designed for point clouds that are sampled from a finite union of smooth sheets and thus is not a general-purpose resolution method for arbitrary non-smooth sets.
A systematic study of the relationship between iterated blow-ups, jet spaces, and higher-order differential invariants such as torsion is left to future work.
\end{remark}

\section{Theoretical Analysis}\label{sec:theoretical-analysis} 

Assuming all intersections are transverse and singularities are thus resolved after one lift, we now consider whether the tangent blow-up achieves two properties needed for geometry processing. First, \emph{separation}: points on distinct sheets that are spatially close must become metrically well-separated after lifting. Second, \emph{smoothness}: the lifted image of a smooth manifold must itself be smooth.

We establish both properties in the continuous setting.  Since we assume that the discrete construction is sampled from this continuous space, these results help justify the computational pipeline. They ensure that, after lifting, each sheet is a smooth submanifold of Euclidean space with positive separation from other sheets, which is the setting where standard convergence results for point cloud Laplacians~\cite{belkin2006convergence,
belkin2009constructing} are formulated.

Our first guarantee is that the lifted metric admits a lower bound between distinct transversely intersecting sheets, even as their spatial distance tends to zero.

\begin{theorem}[Lifted Separation]\label{thm:separation}
Let $M_1, M_2 \subset \mathbb{R}^n$ be smooth $d$-dimensional
submanifolds intersecting at a point $x_0$ with distinct tangent planes.
Write $P_{x_0,M_a}$ for the tangent projector of sheet~$M_a$
at~$x_0$, and let
\[
  \delta \;=\; d_{\mathrm{chord}}\!\bigl(T_{x_0}M_1,\;T_{x_0}M_2\bigr)
  \;>\; 0
\]
be the chordal distance.  Then there exists $\varepsilon_0 > 0$
(depending on~$\delta$ and the curvatures of $M_1,M_2$ near~$x_0$)
such that for all $p \in M_1$, $q \in M_2$ with
$\|p - x_0\|,\,\|q - x_0\| < \varepsilon_0$,
\[
  d_{\mathcal{L}}\!\bigl(\tilde{p},\,\tilde{q}\bigr)
  \;\geq\; \frac{1}{\sqrt{2}}\,\sqrt{\alpha}\;\delta
  \;>\; 0,
\]
where $\tilde{p} = (p,\,T_pM_1)$ and $\tilde{q} = (q,\,T_qM_2)$
are the lifted points and~$d_{\mathcal{L}}$ is the lifted metric
of~\S\ref{sec:lifted-metric}.
\end{theorem}
The proof is given in Appendix A.

We validate the separation bound empirically in Figure~\ref{fig:separation-theorem}.
For the discrete pipeline, the theorem has the following consequence.

\begin{corollary}[Discrete sheet separation]\label{cor:discrete-separation}
Let $\{(x_i, U_i)\}_{i=1}^N$ be a point cloud sampled from $M_1 \cup M_2$ with correct tangent planes, and let $\Phi_i = \bigl(x_i,\,\sqrt{\alpha/2}\;\mathrm{vec}(P_i)\bigr)$ be the lifted embedding.
By Theorem~\ref{thm:separation}, every pair of lifted points
$\Phi_i$,~$\Phi_j$ with $x_i \in M_1$, $x_j \in M_2$, and
$\|x_i - x_0\|,\,\|x_j - x_0\| < \varepsilon_0$ satisfies
\[
  \|\Phi_i - \Phi_j\|
  \;\geq\;
  \tfrac{1}{\sqrt{2}}\,\sqrt{\alpha}\;\delta.
\]
In particular, if $k$-nearest-neighbor queries in the lifted space use a search radius smaller than
$\frac{1}{\sqrt{2}}\sqrt{\alpha}\,\delta$, no query point on one sheet will return a neighbor from the other sheet within
distance~$\varepsilon_0$ of the singularity.
\end{corollary}

\begin{proof}
By construction,
$\|\Phi_i - \Phi_j\|^2
= d_{\mathcal{L}}^2(\tilde{x}_i,\,\tilde{x}_j)$,
so the bound follows directly from Theorem~\ref{thm:separation}.
\end{proof}

Corollary~\ref{cor:discrete-separation} makes the separation theorem practical. Given a non-manifold configuration whose minimal tangent disparity $\delta$ is known (or conservatively estimated), one can choose $\alpha > 2\|\Phi_i - \Phi_j\|^2/\delta^2$ and guarantee that nearest-neighbor queries never cross sheet boundaries near the singularity.

\begin{figure}
  \centering
  \includegraphics[width=\linewidth]{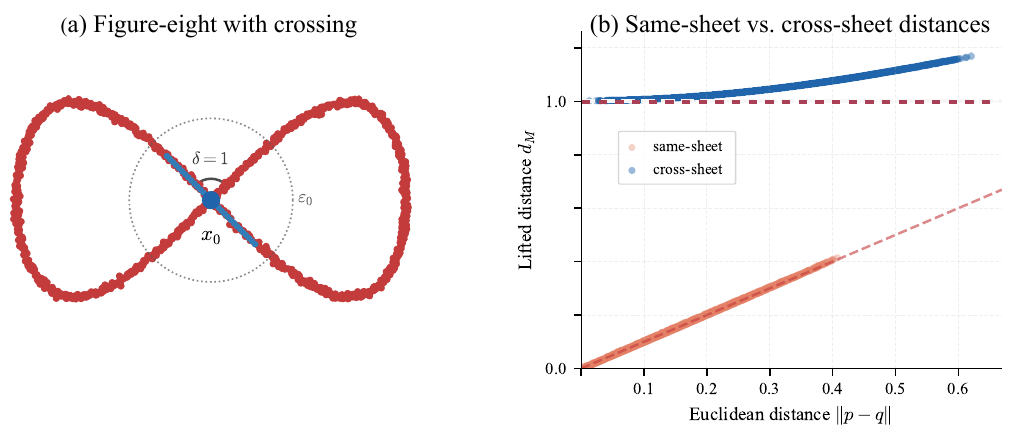}
  \caption{\label{fig:separation-theorem}
    Same-sheet vs.\ cross-sheet lifted distances on a figure-eight curve.  (a)~The curve with crossing point marked.
    (b)~Lifted distance~$d_{\mathcal{L}}$ as a function of Euclidean distance~$\|p - q\|$.  Same-sheet pairs (red) vanish with spatial distance; cross-sheet pairs (blue) are bounded below by the separation floor
    $\frac{1}{\sqrt{2}}\sqrt{\alpha}\,\delta$ (dashed).}
\end{figure}

The second guarantee is that the lifted image of each smooth sheet is regular enough to support the differential operators defined in~\S\ref{sec:lifted-grad-div} and preserves intrinsic dimension. The blow-up differentiates the Gauss map, so each iteration costs exactly one order of differentiability.
\begin{theorem}[Regularity of the Lifted Manifold]\label{thm:regularity} 
Let $M \subset \mathbb{R}^n$ be a $C^r$ submanifold of dimension~$d$ with $r \geq 2$.  Then the lifted embedding $\Phi : M \to \mathbb{R}^{n+n^2}$, $\Phi(x) = \bigl(x,\,\sqrt{\alpha/2}\;\mathrm{vec}(P_x)\bigr)$, is a $C^{r-1}$ embedding, and the lifted manifold $\widehat{M} = \Phi(M)$ is a $C^{r-1}$ submanifold of dimension~$d$.
\end{theorem}
The proof is given in Appendix A.

These results establish that, for transversely intersecting $C^{r}$ manifolds, the blow-up produces a collection of well-separated, $C^{r-1}$-regular submanifolds in Euclidean space. 
        
\section{Lifted Differential Operators}
\begin{figure*}[tbp]
  \centering 
  \includegraphics[width=0.8\linewidth]{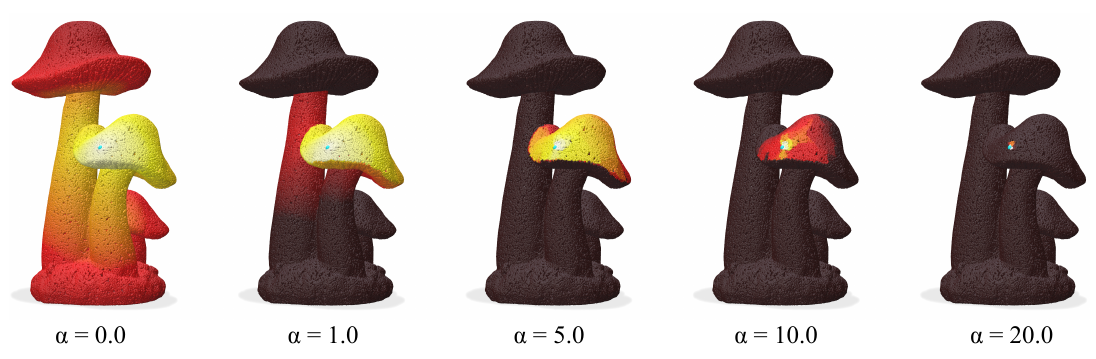}
  \caption{\label{fig:alpha-ablation}Ablation over the projector weight $\alpha$ on the first lift of a mushroom point cloud~\cite{Thingi10K}. Heat is placed at a single source point on the middle cap. At $\alpha=0$ the kernel ignores tangents entirely and heat bleeds across all sheets; as $\alpha$ grows, diffusion remains on each smooth sheet. At higher values of $\alpha$, diffusion becomes highly sensitive to underlying tangents.}
\end{figure*}

To perform downstream geometry processing tasks directly on non-manifold point clouds, we now define examples of discretized differential operators on the tangent blow-up representation following standard meshless constructions. We construct a point cloud Laplacian for diffusion and spectral analysis (\S\ref{sec:lifted-laplacian}), gradient and divergence operators for vector field computations (\S\ref{sec:lifted-grad-div}), and define pointwise curvature invariants extracted from the second fundamental form (\S\ref{sec:lifted-curvature}).

\subsection{Laplacian}\label{sec:lifted-laplacian}

We model the lifted point cloud as a weighted graph by using a multi-scale product kernel derived from the lifted product metric that decouples spatial and tangential components:
\[
    W_{ij} = \exp\!\left(-\frac{\|x_i - x_j\|^2}{\sigma_0^2}\right) \cdot \prod_{m=1}^M \exp\!\left(-\frac{\|P_i^{(m)} - P_j^{(m)}\|_F^2}{\sigma_m^2}\right),
\]
where $\sigma_0$ is the spatial bandwidth and $\sigma_m$ is the bandwidth for the $m$-th tangent projector given $M$ blow-up iterations. 
Each bandwidth can either be fixed, or set to the median $k$-nearest-neighbor distance in the corresponding factor of the product metric~\cite{zelnik2004self}, so that the kernel adapts to the intrinsic scale of each component.
Incorporating Grassmannian components results in geometrically distinct points with differing unoriented tangents having reduced affinity near singularities.

From the affinity $W$ and degree matrix $\operatorname{diag}(W\mathbf{1})$ we form either the standard unnormalized Laplacian $L=\operatorname{diag}(W\mathbf{1})-W$ or its symmetric normalized variant $L_{\mathrm{sym}}=I-\operatorname{diag}(W\mathbf{1})^{-1/2}W\operatorname{diag}(W\mathbf{1})^{-1/2}$~\cite{ coifman2006diffusion,von2007tutorial}.
Recall from \S\ref{sec:theoretical-analysis} that transversely intersecting sheets are pulled apart in the lifted embedding by their tangent spaces, while each smooth sheet remains smooth under the lift. Thus, if the point cloud sample is drawn from $s$ smooth components whose pairwise tangent angles exceed the kernel bandwidth, for sufficiently fine sampling the lifted affinity is close to block diagonal, with one block per component. Standard perturbation theory for block-diagonal symmetric matrices~\cite{von2007tutorial} yields exactly $s$ near-zero eigenvalues whose corresponding eigenvectors approximate indicators of the individual components, which are then recovered via the spectral clustering pipeline of \S\ref{sec:spectral-segmentation}.

\subsection{Gradient and Divergence}\label{sec:lifted-grad-div}

We construct standalone gradient and divergence operators using a standard weighted moving least squares (MLS) construction on the lifted point cloud~\cite{nealen2004short, liang2013solving,wiersma2022deltaconv}. Both operators reuse the self-tuning product-kernel weights $W_{ij}$ from the affinity matrix $W$ defined in \S\ref{sec:lifted-laplacian}.
At each point $\Phi_i$ the tangent projector $Q_i = R_i R_i^\top$ provides an orthonormal tangent frame $R_i\in\mathbb{R}^{D\times d}$, where $D = n + n^2$ is the ambient dimension of the first lift and $d$ is the rank of the tangent projector $Q_i$.
For a second-level lift, $\Phi_i$, $R_i$, and $Q_i$ are replaced by their level-two counterparts from Algorithm \ref{alg:iterated-blowup}, with ambient dimension $D_2 = D_1 + D_1^2$.

\paragraph*{Discrete gradient.}
Let $f_i = f(\Phi_i)$ be a scalar function on the lifted point cloud. 
We estimate $\operatorname{grad} f_i$ by fitting a degree-one polynomial in the tangent plane at point $\Phi_i$.  
First, project each neighbor displacement into the local tangent frame,
\(
  \bar\delta_{ij} \;=\; R_i^\top(\Phi_j - \Phi_i)
    \;\in\; \mathbb{R}^d,
\)
and then form the $d\times d$ weighted moment matrix
\[
  S_i \;=\; \sum_{j\in\mathcal{N}(i)}
    W_{ij}\,\bar\delta_{ij}\,\bar\delta_{ij}^\top.
\]
Substituting $f(\Phi_j) \approx f_i + \bar g^\top \bar\delta_{ij}$ and minimizing $\sum_j W_{ij}\,\|\bar g^\top \bar\delta_{ij} - (f_j - f_i)\|^2$ over $\bar g \in \mathbb{R}^d$ yields
\[
  \bar g_i \;=\; S_i^{-1} \sum_{j\in\mathcal{N}(i)} W_{ij}\,\bar\delta_{ij}\,(f_j - f_i).
\]
Lifting back to the ambient space gives the estimated gradient:
\(
  (\operatorname{grad}_{\widehat{\mathbb{Y}}}\, f)_i \;=\; R_i\,\bar g_i \;\in\; \mathbb{R}^{D}.
\)

\paragraph*{Discrete divergence.} Given a tangent vector field $X_i \in \mathrm{col}(R_i)$ on the lifted manifold, we estimate its divergence at vertex $i$ by regressing changes in $X$ against displacements, with all quantities expressed in the local tangent frame $R_i$. Concretely: \begin{enumerate} 
\item Project the vector field at $i$ and at each neighbor $j$ into point $i$'s frame: $\bar{X}_k^{(i)} = R_i^\top X_k$ for $k \in {i} \cup \mathcal{N}(i)$. 
\item Estimate the tangent-plane Jacobian by regressing the frame-consistent differences against projected displacements: $$J_i^{a} = \sum_{j \in \mathcal{N}(i)} W_{ij} \, \bar{\delta}_{ij} \, \bigl(\bar{X}_j^{(i)} - \bar{X}_i^{(i)}\bigr)^a \; \in \; \mathbb{R}^d.$$
\item Take the trace: $(\mathrm{div}_{\widehat{\mathbb{Y}}}\, X)_i = \mathrm{tr}(S_i^{-1}\, J_i)$, where $J_i = [J_i^1 ;\cdots; J_i^d]$. 
\end{enumerate} 

\subsection{Curvature}\label{sec:lifted-curvature}

Having estimated $B_i \approx B_{x_i}$ at each sample (\S\ref{sec:second-lift}), we now extract classical curvature invariants. We restrict to hypersurfaces ($n - d = 1$) in this section; the estimation of $B_i$ applies in any codimension, but the invariants extracted differ.

For a hypersurface, $N_i$ is a single unit normal and $B_i(e_a) \in \mathbb{R}^{1 \times d}$ is a row vector. Stacking the rows gives the $d \times d$ symmetric matrix $S_i = [B_i(e_1);\,\cdots\,;\,B_i(e_d)]$, which is the shape operator in the tangent basis $U_i$. Its eigenvalues are the principal curvatures $\kappa_1, \ldots, \kappa_d$, from which we compute:
\begin{align*}
  K &= \textstyle\prod_a \kappa_a = \det(S_i),\\
  H &= \textstyle\frac{1}{d}\sum_a \kappa_a
     = \frac{1}{d}\,\mathrm{tr}(S_i),\\
  \|\mathrm{II}\|_F &= \textstyle\sqrt{\sum_a \kappa_a^2}
     = \|S_i\|_F.
\end{align*}
Although $S_i$ depends on the choice of $U_i$ and $N_i$, these invariants do not. $K$ and $\|\mathrm{II}\|_F$ are unchanged under any change of basis, and $H$ changes only in sign under $N_i \to -N_i$. We therefore report the sign-free quantities $K$, $H^2$, and $\|\mathrm{II}\|_F$, which require no consistent orientation (we discuss the implications of orientation invariance in \S\ref{sec:unoriented-discussion}).

\section{Experiments}

We demonstrate the tangent blow-up on four tasks: geodesic computation (\S\ref{sec:lifted-geodesics}), segmentation (\S\ref{sec:spectral-segmentation}), surface parameterization (\S\ref{sec:spectral-parameterization}), and curvature estimation (\S\ref{sec:lifted-curvature}).
We set $\alpha = 1.0$ for all examples, except in Figure~\ref{fig:teaser}, where we use $\alpha = 2.0$. For transverse intersections, $\alpha$ can be chosen using the separation bound in Theorem~\ref{thm:separation}. In principle, the same criterion applies recursively to higher-order lifts. In practice, increasing $\alpha$ makes the metric more sensitive to curvature-induced tangent variation, as illustrated in Figure~\ref{fig:alpha-ablation}. For iterated blow-ups, this applies at each level, with higher-level weights controlling sensitivity to higher-order features.

For kernel constructions, we use $k = 20$ nearest neighbors.
Spatial and angular bandwidths are set via self-tuning median $k$-NN distances~\cite{zelnik2004self}, where $\sigma_0$ is the median $k$-NN distance in spatial coordinates across all points, and each angular bandwidth $\sigma_{m}$ is the median $k$-NN distance in the corresponding projector component of the lifted product metric defined in \S\ref{sec:lifted-metric}. For geodesic computation, these product-kernel weights also serve as the regression weights $W_{ij}$ for the lifted gradient and divergence operators.

We test our algorithm on four classes of data. \emph{Analytic shapes} like the Klein bottle provide controlled ground truth points and normals for quantitative evaluation~\cite{do2016differential}. 
The \emph{Thingi10k}~\cite{Thingi10K} and \emph{ThreeDScans}~\cite{threedscans} datasets provide non-manifold geometry with self-intersections and T-junctions; we uniformly sample points from these meshes and take face normals.
Lastly, the \emph{SceneNN}~\cite{hua2016scenenn} dataset provides real scans of rooms for testing our method on noisy data with estimated normals.

\subsection{Geodesics}\label{sec:lifted-geodesics}
We compute approximate lifted geodesic distances via the heat method~\cite{crane2013geodesics}, which recovers geodesic distance in three steps: (1) diffuse a delta source for short time by solving $(I + tL)u = \delta_s$, (2) normalize the gradient of the diffused field to obtain a unit vector field $X = -\nabla u / |\nabla u|$ pointing away from the source, and (3) solve a Poisson equation $L\phi = -\mathrm{div}(X)$ to recover the distance function. In the lifted setting, we replace the standard Laplacian, gradient, and divergence with their lifted counterparts.

We compare against the nonmanifold Laplacian of Sharp and Crane~\cite{sharp2020laplacian}, which constructs an intrinsic Delaunay triangulation that handles non-manifold edges and boundaries. It operates on positions only and does not use tangent-plane information for sheet separation.
The distances recovered by the lifted heat method approximate geodesics in the product metric, which penalizes tangent-plane rotation proportionally to $\alpha$. Hence, as seen in Figure~\ref{fig:geodesics}, lifted geodesics tend to favor straighter paths with smoothly varying tangents. This notion of distance is applicable whenever one wants shortest paths that respect the sheet structure of a non-manifold surface -- for example, tracing paths on individual components of an object like the ship in Figure~\ref{fig:geodesics}, without segmenting the hull, deck, and sails beforehand.

\begin{figure}
    \centering
    \includegraphics[width=\linewidth]{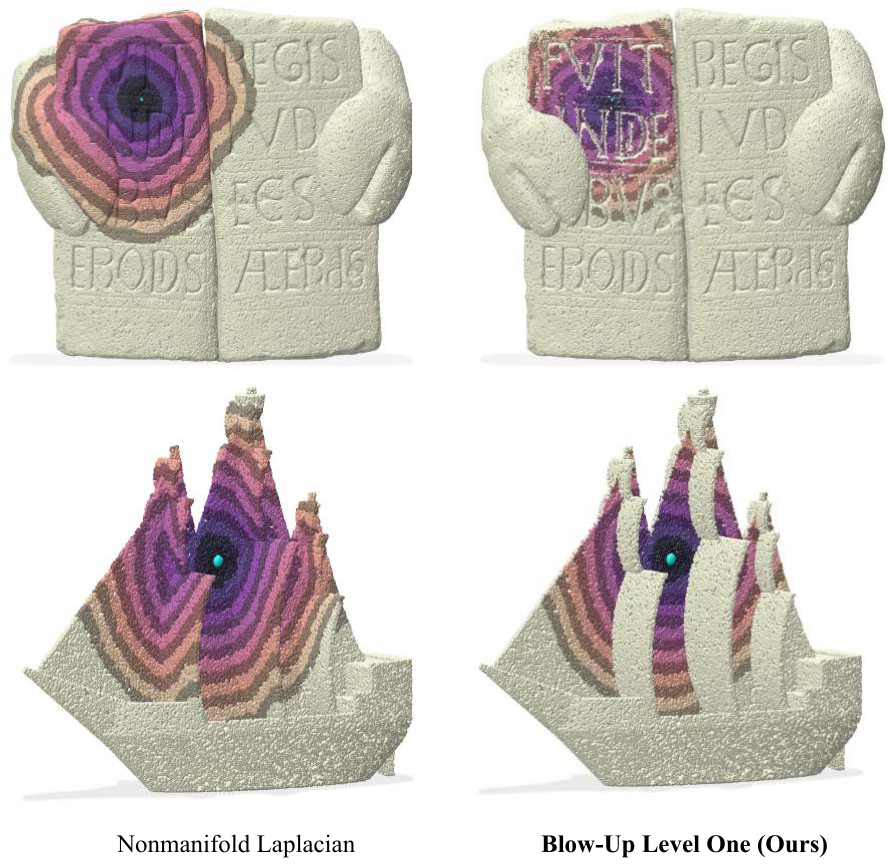}
    \caption{\label{fig:geodesics} Geodesic distance via the heat method, with the source point marked in cyan. The nonmanifold Laplacian~\cite{sharp2020laplacian} (left) diffuses across self-intersections, while the lifted Laplacian (right) is restricted to geometrically distinct components.}
\end{figure}

To demonstrate our approach on real scanned data with estimated normals, we compute geodesics for several rooms from the \textit{SceneNN} dataset \cite{hua2016scenenn}. Any normal estimation algorithm can be used; we opt for \textit{HSurf-Net} \cite{li2022hsurf}. The results are illustrated in Figure~\ref{fig:rooms}, where we compute lifted geodesics along scanned floors. Despite missing regions, clutter, and noise, the resulting distance contours are constrained to the floor surface and respect obstacles such as walls and furniture. Such behavior may be useful for downstream navigation tasks such as path planning, where lifted geodesics could allow a robot to reason about shortest traversable paths along the floor while avoiding non-traversable regions.

\begin{figure}
    \centering
    \includegraphics[width=\linewidth]{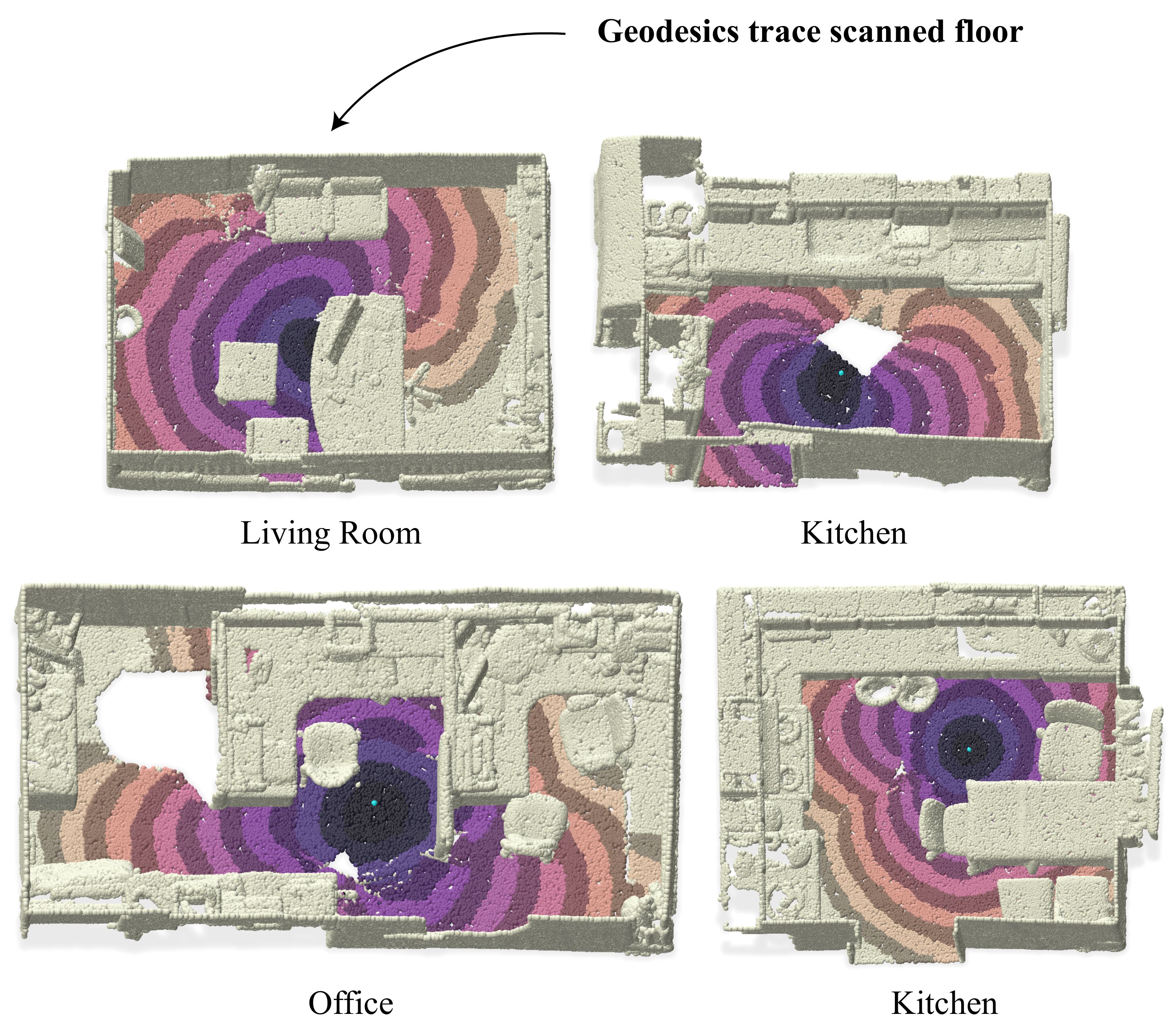}
    \caption{\label{fig:rooms}
    Lifted geodesics on scanned indoor scenes from the \textit{SceneNN} dataset~\cite{hua2016scenenn} using normals estimated by \textit{HSurf-Net}~\cite{li2022hsurf}. Across living room, kitchen, and office scans, the computed distance contours trace the floor surface while respecting walls, furniture, and missing regions. This suggests that our method can operate on noisy real-world reconstructions and is relevant for applications such as path planning on scanned indoor environments.
    }
\end{figure}

\subsection{Spectral Segmentation}\label{sec:spectral-segmentation}
The block structure of $L_{\mathrm{sym}}$ established in
\S{\ref{sec:lifted-laplacian}} reduces segmentation of a non-manifold point cloud to standard spectral clustering~\cite{ng2001spectral,von2007tutorial}. We take the $m = 10$ smallest non-trivial eigenvectors of $L_{\mathrm{sym}}$, stack them as rows, and then unit-normalise each row~\cite{ng2001spectral}. We then use DBSCAN~\cite{ester1996density} with radius $\varepsilon=0.2$ to cluster points, which avoids fixing the number of components and, because it ignores directions of low variance, is insensitive to moderate overestimates of $m$.

As seen in Figure~\ref{fig:segmentation}, lifted segmentation recovers geometrically distinct components such as individual icicles or the distinct sides of the fandisk.
Since neighboring points on different components, as occur near sharp corners or sheet intersections, carry distinct projectors, the lifted metric suppresses their mutual affinity in $W$ without any explicit singularity detection or feature classification.
Recovering individual smooth components from non-manifold geometry in such a way can be useful for downstream tasks such as per-component surface parameterization (\S\ref{sec:spectral-parameterization}) and curvature estimation (\S\ref{sec:curvature}).

\begin{figure}[h]
    \centering
    \includegraphics[width=\linewidth]{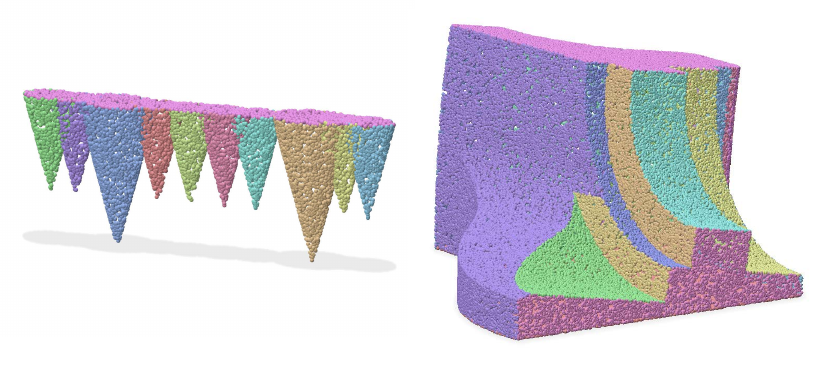}
    \caption{\label{fig:segmentation}Lifted spectral segmentation on two non-manifold point clouds.  Left: the self-intersecting \emph{icicles} model (Thingi10k~\cite{Thingi10K}). Right: the piece-wise smooth \emph{fandisk}~\cite{hoppe1994piecewise}. Lifted spectral segmentation recovers individual smooth components.}
\end{figure}

\paragraph*{Spectral Parameterization.}\label{sec:spectral-parameterization}

Once the segmentation of \S{\ref{sec:spectral-segmentation}} has isolated smooth components, we parameterize them individually using standard Laplacian eigenmaps~\cite{belkin2003laplacian} on the same lifted Laplacian used previously (the self-tuning product-kernel, symmetric normalised $L_{\mathrm{sym}}$ from \S{\ref{sec:lifted-laplacian}}).

We take the two smallest non-trivial eigenvectors of $L_{\mathrm{sym}}$ and use them directly as UV coordinates.
Because $L_{\mathrm{sym}}$ is defined using the lifted metric, the corresponding eigenmap respects the local tangent frame even when the underlying components are self-intersecting and nonorientable; Figure~\ref{fig:nonorientable_parameterization} illustrates this on two nonorientable surfaces in $\mathbb{R}^3$ whose immersions necessarily self-intersect. The lifted Laplacian separates intersecting components through the projector term in the product metric, and because the projector $P = UU^\top$ is sign-free by construction, the eigenmaps extend smoothly across regions where a consistent normal choice is impossible.

Typically, to perform surface parameterization, an artist would need to cut a shape into individual pieces that are separately parameterized and combined into a texture atlas. Our method resolves intersections and produces geometric cuts automatically, since the spectral gap induced by differing tangent planes (or, at the second blow-up level, differing curvatures) partitions the surface into smooth components. This means that an artist or automated pipeline can go from a non-manifold point cloud to per-component processing without ever manually cutting the geometry apart.

\begin{figure}
    \centering
    \includegraphics[width=\linewidth]{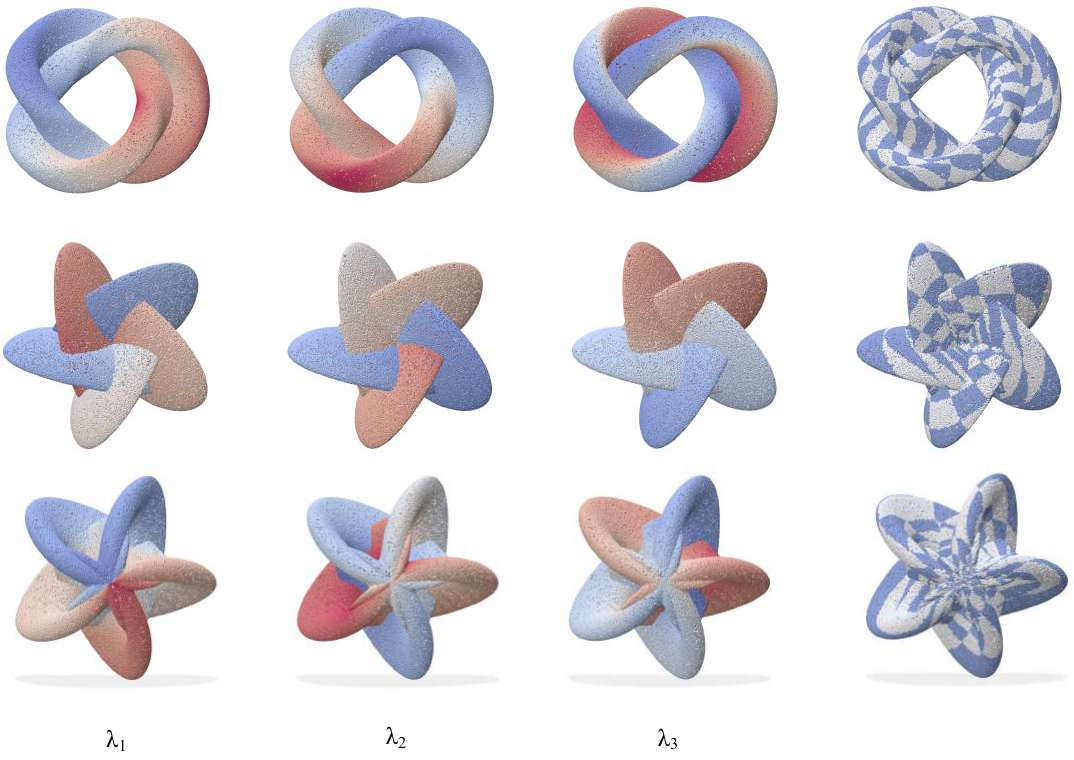}
    \caption{\label{fig:nonorientable_parameterization}Spectral parameterization of two self-intersecting nonorientable surfaces: Banchoff's Klein bottle~\cite{banchoff1976minimal,Busser-nonorientable} (top) and a five-branch Boy-like surface~\cite{petit1981representation,Busser-nonorientable} (bottom two rows). Left to right:  of the three smallest non-trivial eigenvectors of the lifted Laplacian $L_{\mathrm{sym}}$, followed by the induced UV checker overlay. Because the lifted Laplacian is built from projectors, the eigenmaps extend smoothly through self-intersections and across regions where a consistent normal choice is impossible.}
\end{figure}

\subsection{Curvature}\label{sec:curvature}

For a parametric immersion $f\colon U \to \mathbb{R}^3$ whose Jacobian has full rank everywhere, Gaussian curvature~$K$ and mean curvature~$H$ are well-defined at every parameter value, even where the image surface self-intersects in $\mathbb{R}^3$~\cite{do2016differential}. A point cloud sampled from the image of~$f$ thus inherits per-point parameter coordinates $(u_i, v_i)$ and well-defined curvature per sheet. 

The lifted regression of $B_i$ (\S\ref{sec:lifted-curvature}) approximates the curvature of the underlying immersion at transverse intersections. We use this per-sheet ground truth for quantitative evaluation on two parametric surfaces and complement it with qualitative comparisons on Thingi10k models~\cite{Thingi10K} and ThreeDScans~\cite{threedscans}, all sampled at up to $N = 200{,}000$ points.

\paragraph*{Baselines.}
We compare against two point cloud curvature estimators.
Jet fitting~\cite{CazalsPouget2005} estimates a
truncated Taylor expansion (an osculating jet) of the surface in a local coordinate system aligned with the estimated normal. We use a degree-$2$ jet with $k = 20$ neighbors to match our blow-up baseline. Randomized Corrected Curvature Measures (CNC)~\cite{LachaudCoeurjolly2023} estimates the curvature tensor via corrected curvature measures~\cite{LachaudRomonThibert2022}: for each query point, $L$ random triangles are sampled from the $k$ nearest neighbors. We use the authors' code with $k = 20$ and $L = 100$. Jet fitting requires only positions and estimates its own local frame from the $k$-NN neighborhood; CNC additionally requires oriented normals as input.

\subsection{Quantitative evaluation.}
We consider two surfaces: the torus, a smooth manifold that serves to validate our implementation, and the Klein bottle, a nonorientable surface whose immersion in $\mathbb{R}^3$ contains self-intersections.

The analytic curvature of the torus depend only on the poloidal angle~$v$~\cite{robbin2022introduction}, so we can assess estimation accuracy as a function of a single variable. Figure~\ref{fig:torus-curvature-plot} shows the estimated
curvature profiles alongside the analytic ground truth (top row) and the corresponding absolute errors (middle row) as functions of the poloidal angle~$v$. All three estimators track the ground truth closely, with median errors in the range $10^{-4}$ to  $10^{-3}$
The dip near $v = \pi$ in the $H^2$ error reflects the fact that $H$ passes through zero at the inner equator for these parameter values ($R = 2r$), so the absolute error is small.

The Klein bottle is nonorientable and self-intersects when immersed in $\mathbb{R}^3$~\cite{abbena2017modern}. To quantify how each method degrades near singularities, we estimate each sample point's proximity to the Klein bottle's self-intersection. Two points that are close in $\mathbb{R}^3$ but far apart in the parameter domain $(u,v)$ must lie on different sheets of the immersion. For each sample, we find its nearest neighbor in $\mathbb{R}^3$ among all samples whose parameter-space distance on the periodic $[0,2\pi)^2$ domain exceeds the heuristic threshold $\pi/2$, and record the ambient distance to that neighbor. We call this the \emph{singularity distance} of that sample.

\begin{figure}[ht]
    \centering
    \includegraphics[width=\linewidth]{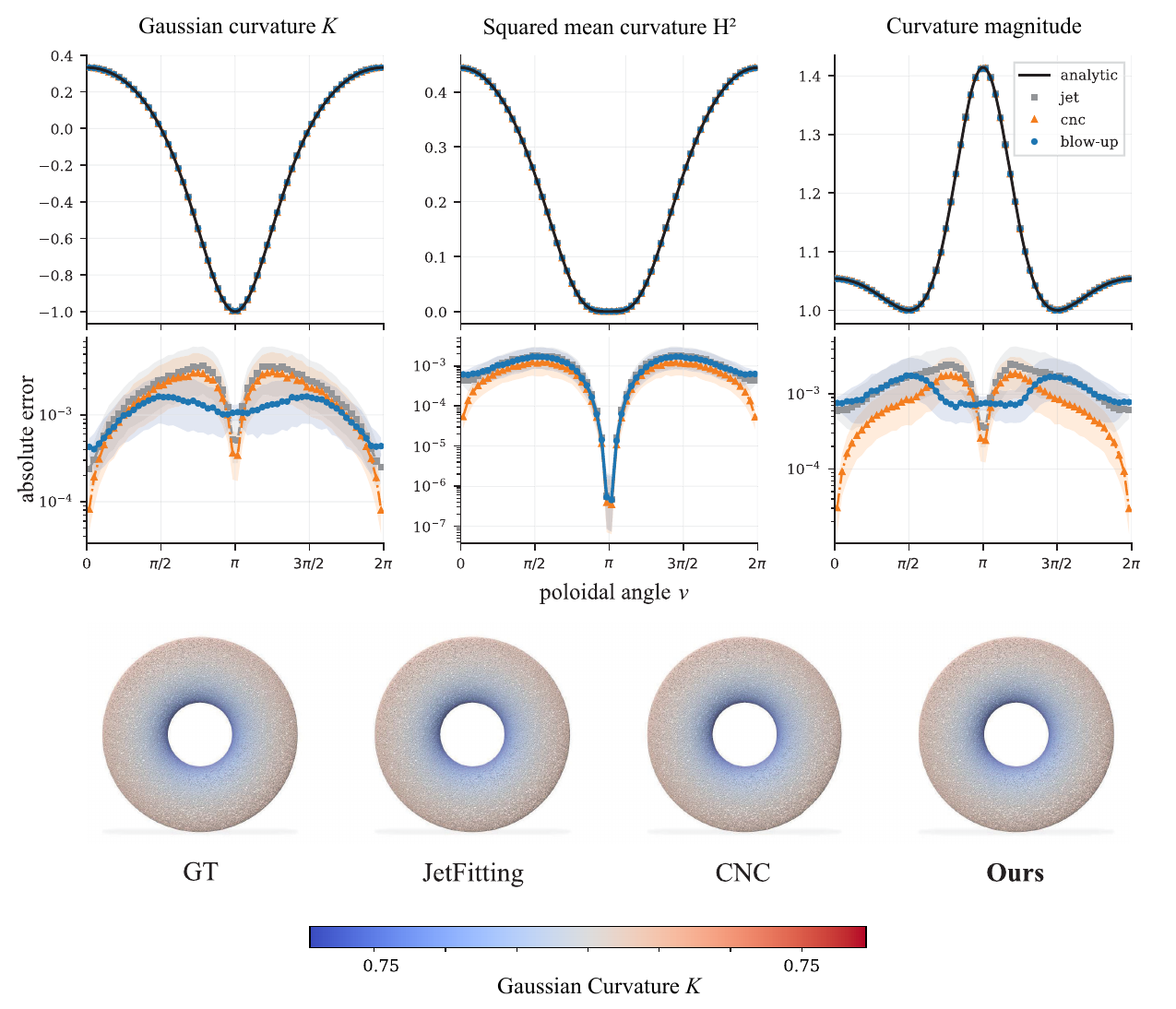}
    \caption{\label{fig:torus-curvature-plot}
    Curvature estimation on the torus ($R=2$, $r=1$,
    $N=100{,}000$), a smooth manifold with no singularities. Top row: estimated $K$ and $H^2$ as a function of the poloidal angle~$v$; the analytic ground truth (black) is nearly occluded by all three estimators. Middle row: median absolute error. All three methods achieve comparable accuracy. Bottom row: Gaussian curvature rendered on the surface.}
\end{figure}

Figure~\ref{fig:error-vs-singularity} plots the median absolute error as a function of singularity distance. Near the self-intersection (singularity distance ${<}\,0.1$), both jet fitting and CNC exhibit error spikes of two to three orders of magnitude: median $|K - K_{\mathrm{gt}}|$ rises from ${\sim}10^{-3}$ to ${\sim}10^{1}$ for jet fitting, and CNC rises similarly. The Klein bottle's nonorientability compounds the problem for CNC, which requires oriented normals as input.
The blow-up error, by contrast, remains flat at
${\sim}5 \times 10^{-4}$ across all singularity distances. Figure~\ref{fig:klein-curvature} renders the three curvature invariants on the surface, with ground truth for comparison.

\begin{figure}[t]
    \centering
    \includegraphics[width=\linewidth]{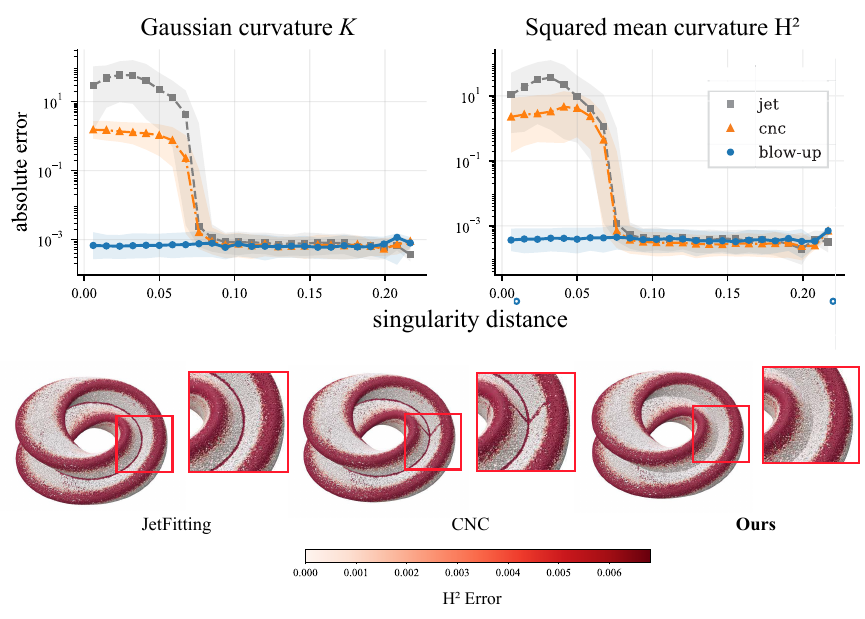}
    \caption{\label{fig:error-vs-singularity}
    Curvature estimation error on the Klein bottle with $R = 2$ and $(u,v) \in [0,2\pi)^2$, sampled at $N = 100{,}000$ points, as a function of singularity distance (defined in \S\ref{sec:curvature}). Each sample is binned by its ambient distance to the nearest cross-sheet neighbor (25 bins, median point plotted with interquartile range shaded). Bottom row: per-point error in curvature magnitude $|\|I\!I\|_F - \|II_{\mathrm{gt}}\|_F|$ rendered on the surface for each method; the blow-up shows no visible degradation at the intersection.}
\end{figure}

\begin{figure}
    \centering
    \includegraphics[width=\linewidth]{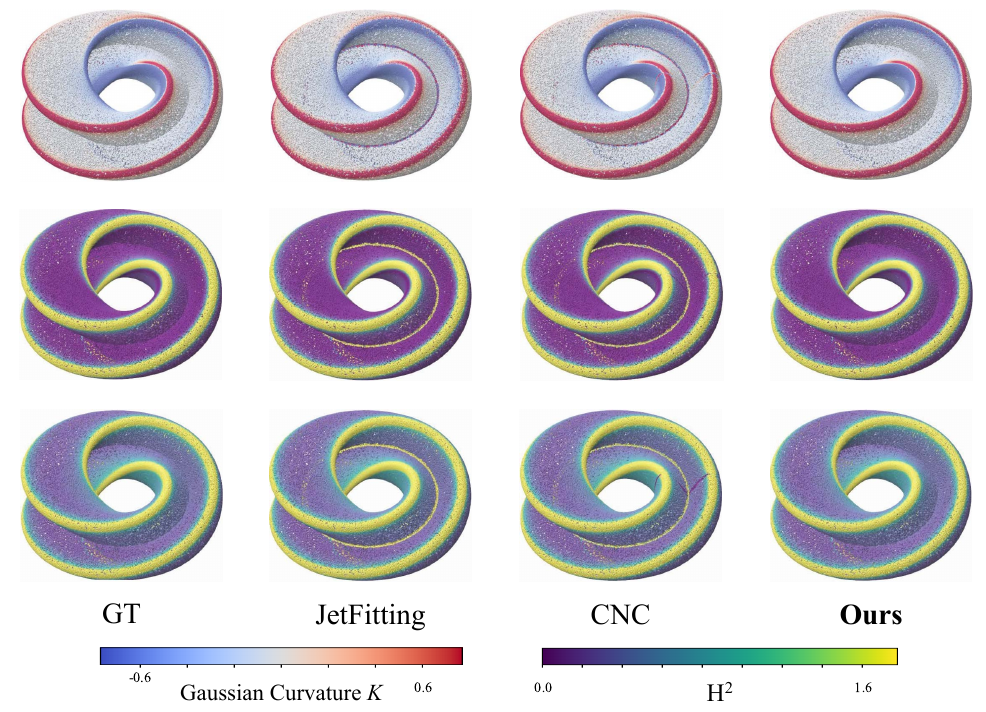}
    \caption{\label{fig:klein-curvature}
    Curvature estimates on the Klein bottle (left to right: ground truth, jet fitting, CNC, blow-up). Rows show $K$ (top), $H^2$ (middle), and curvature magnitude $\|I\!I\|_F$ (bottom). Jet fitting and CNC produce visible artifacts along the self-intersection circle, while the blow-up closely matches the ground truth across all three invariants.}
\end{figure}

\subsection{Qualitative evaluation.}
We run all three methods on point clouds sampled from several meshes that contain self-intersections, T-junctions, or thin features. Each model is uniformly sampled at up to $200{,}000$ oriented points. No analytic ground truth is available, so our criterion is whether the estimated curvature field is ``spatially coherent'' on each geometric component, meaning smooth where the surface is smooth, and free of high curvature artifacts at self-intersections. All three methods use $k = 20$ nearest neighbors; jet fitting uses degree~$2$, CNC uses $L = 100$ random triangles with a uniform kernel.

The torus knot (Figure~\ref{fig:knot-gaussian}) is a mesh from
Thingi10k~\cite{Thingi10K} whose geometry
passes over and under itself. We render its Gaussian curvature~$K$. The inset (red box) zooms into a self-intersecting region. CNC (left) produces a curvature field such that the local sign of $K$ becomes negative at intersections. The blow-up (right) preserves positivity through every crossing, meaning the curvature on each tube segment is less affected by the presence of the other.

\begin{figure}
    \centering
    \includegraphics[width=\linewidth]{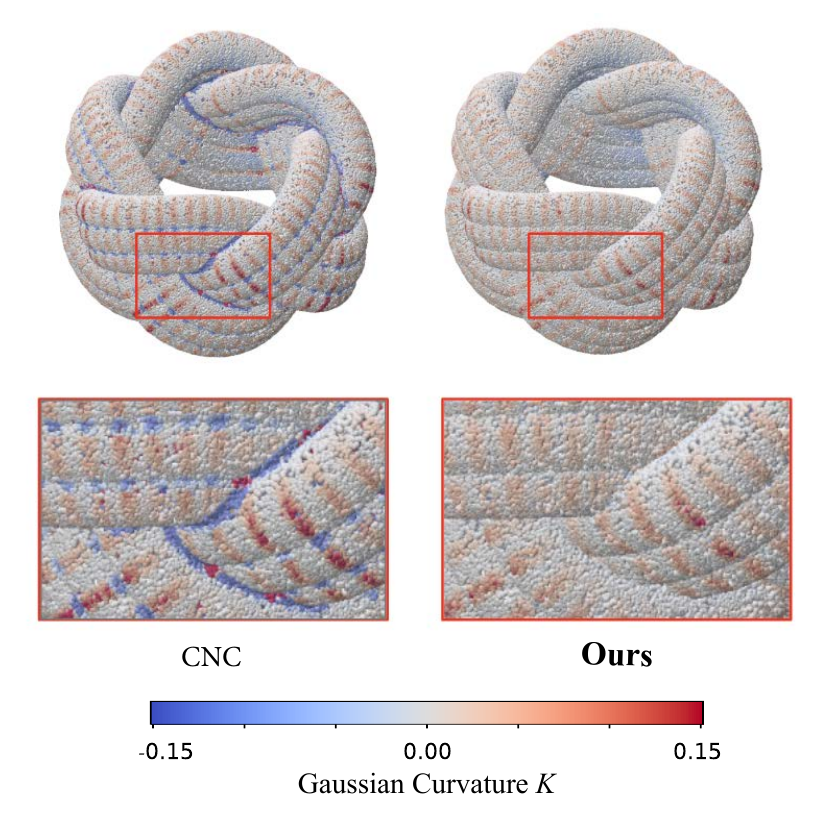}
    \caption{\label{fig:knot-gaussian} Gaussian curvature $K$ on a self-intersecting torus knot.
    Bottom: inset zoom of the boxed crossing region. CNC produces negative curvature at the crossings. The blow-up maintains consistent banding on each knot segment through the intersections. Positive red banding is recovered from the underlying mesh.}
\end{figure}

The propeller model~\cite{Thingi10K} (Figure~\ref{fig:barspin-h2}) consists of blades which intersect a cylindrical hub. The blow-up recovers geometrically meaningful curvature on each separate component: the cylindrical base carries uniform magnitude and the flat blade faces read near zero. Jet fitting and CNC both report elevated $H^2$ along every intersection seam and sharp edge, making it difficult to distinguish the curvature of individual sheets from non-manifold junctions.

Figure~\ref{fig:greyounds-curvature} shows squared mean curvature $H^2$ estimated on a point cloud sampled from two greyhound sculptures (ThreeDScans~\cite{threedscans}). Where the bodies touch, CNC reports elevated curvature along the intersection curves, visible as yellow streaks across the paws and chest in the close-up. The blow-up estimator attributes curvature to each sheet independently, so these ridges are absent. In practice this means per-component curvature fields can be recovered directly from the raw point cloud, with no prior segmentation or intersection removal.

\begin{figure}
    \centering
    \includegraphics[width=\linewidth]{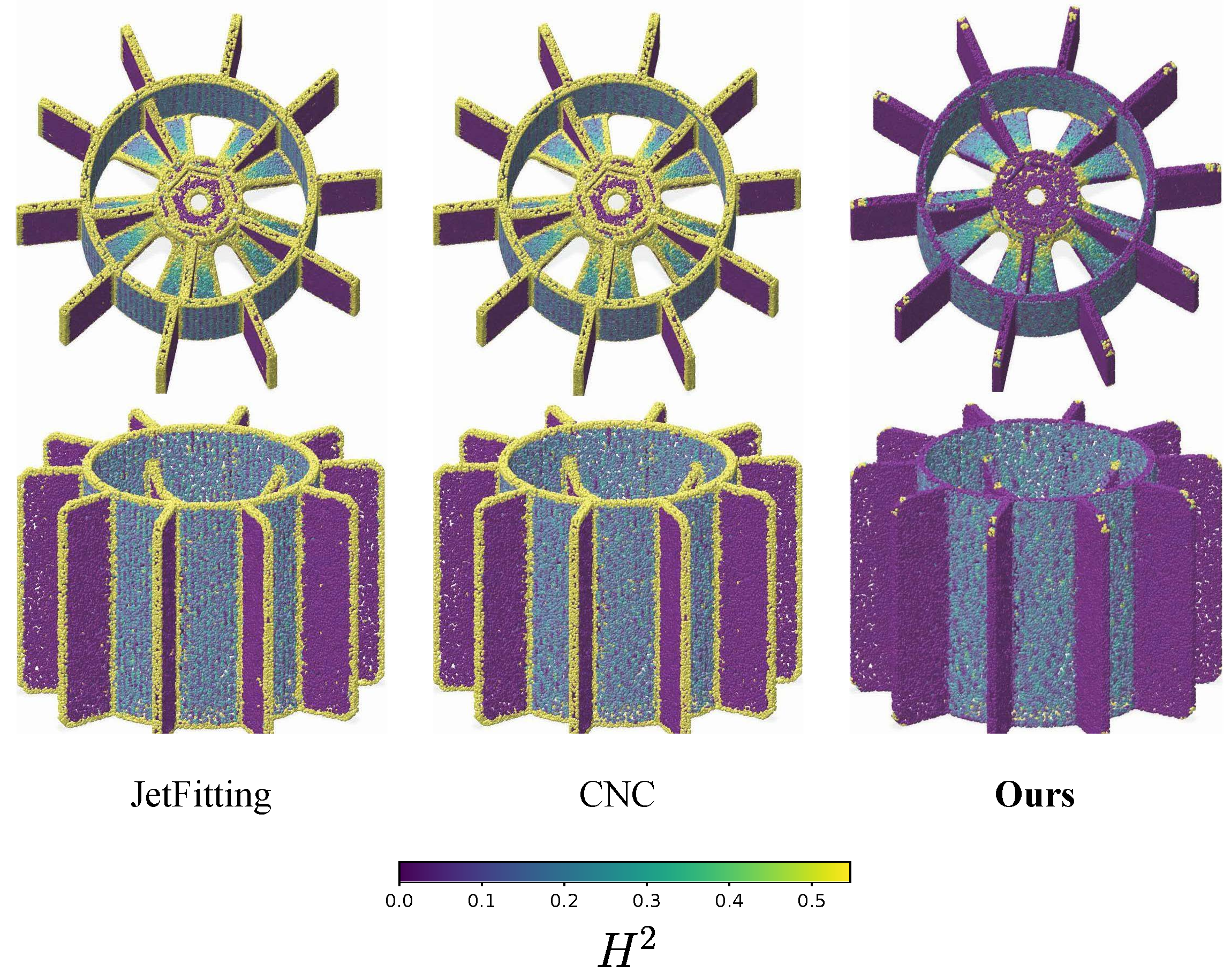}
    \caption{Mean squared curvature $H^2$ on a self-intersecting propeller model (Thingi10k~\cite{Thingi10K}). Left: jet fitting. Center: CNC. Right: blow-up (ours). Jet fitting and CNC exhibit high-magnitude estimates (yellow) along intersections and edges. The discrete blow-up recovers consistent per-component curvature: the cylindrical base carries uniform magnitude, and the flat blade faces read near zero. No singularity detection or segmentation is applied as a preprocessing step.}
    \label{fig:barspin-h2}
\end{figure}

\paragraph*{Orientation Invariance.}\label{sec:unoriented-discussion}
The curvature invariants we report are orientation-free by definition. Hence, any method that estimates them correctly should produce the same values regardless of a globally consistent normal field. 
Jet fitting, as seen in Figure~\ref{fig:klein-curvature}, is an ``orientation free'' algorithm~\cite{CazalsPouget2005}. By contrast, CNC uses oriented normals, so inconsistent orientations corrupt the estimates~\cite{LachaudCoeurjolly2023}.
Our algorithm avoids oriented normals entirely, since the projector $P = UU^\top$, its derivative $\mathcal{P}_v$, and the regression for $B_i$ all operate without choosing a consistent normal direction. Orientation enters only in the final step when extracting $B_x(v) = N^\top \mathcal{P}_v\,U$, where we report orientation-free invariants. This makes the estimator applicable to nonorientable surfaces (Figure~\ref{fig:mobius-intro}), but also means curvature is underestimated at thin shell edges where opposite faces carry nearly parallel tangent planes with opposite normals. The projector sees both ``sides'' of a thin sheet as having the same tangent plane, so it reports low curvature at thin edges whereas CNC detects the normal reversal, as seen in Figure~\ref{fig:c6-curvature}. Whether this trade-off is desirable depends on the application.

\begin{figure}[ht]
    \centering
    \includegraphics[width=\linewidth]{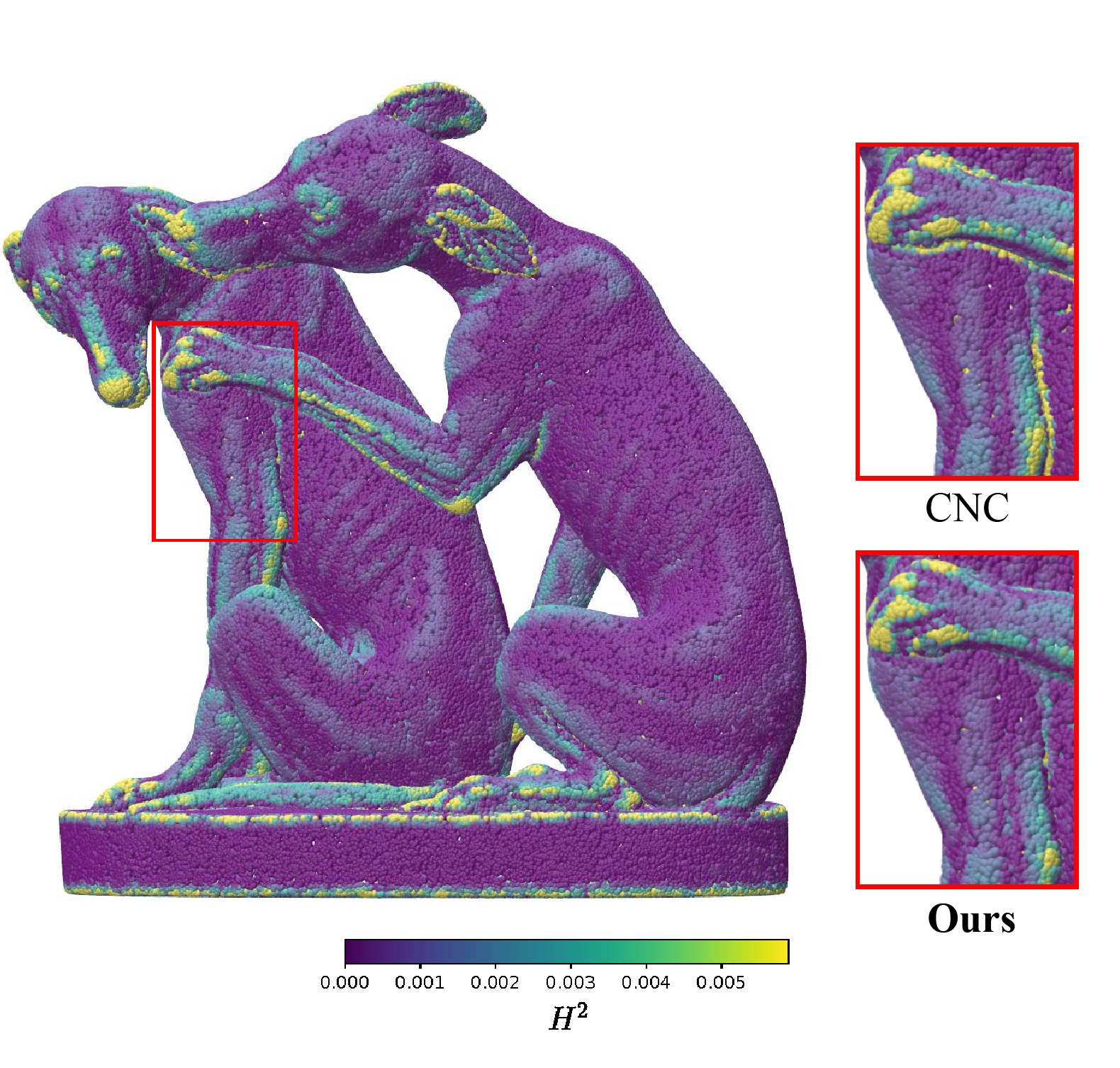}
    \caption{\label{fig:greyounds-curvature}Squared mean curvature $H^2$ on two intersecting greyhound sculptures (ThreeDScans~\cite{threedscans}). Top: blow-up (ours). Bottom: CNC reports elevated curvature along intersection curves (visible as yellow ridges along the paws and legs); the blow-up attributes curvature to each greyhound independently.}
\end{figure}

\begin{figure}
    \centering
    \includegraphics[width=\linewidth]{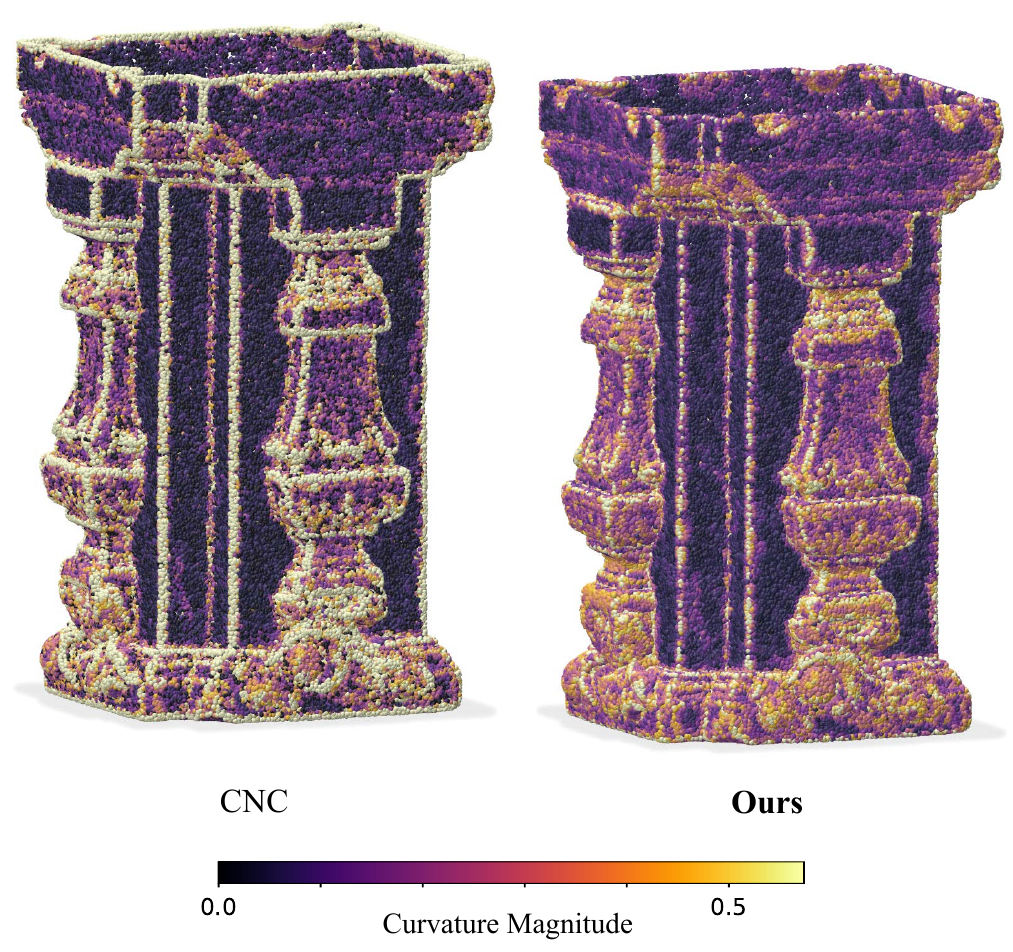}
    \caption{\label{fig:c6-curvature} Curvature magnitude $\|I\!I\|_F$ on a column model (ThreeDScans~\cite{threedscans}) with thin walls. Left: CNC reports high curvature along thin edges where oriented normals reverse between opposite faces. Right: the blow-up estimator reports low curvature at these edges because the tangent plane itself has little rotation.}
\end{figure}

\section{Discussion}
Our formulation assumes that a tangent plane is available at each sample. Future work could address automatically estimating multiple tangents at each point in the point cloud. For example, near singularities, PCA-based tangent estimation averages across intersecting components and can be sensitive to noise. 
This is illustrated in Figure~\ref{fig:tangent-noise}, where the lifted heat method is robust to positional noise but sensitive to estimated tangents.

More specifically, we find that our method is generally robust to positional noise if one has reliable normals, as can be the case in data obtained by photometric scanning, but degrades with respect to increasingly noisy normals. Figure~\ref{fig:noise-study} illustrates this effect, where we have taken a point cloud sampled from a propeller mesh \cite{Thingi10K} and used local PCA to estimate normals. By adding noise to the positions and estimated tangents separately, we see how the lifted heat method is affected: lifted components remain well separated with respect to positional noise, but degrade more rapidly in correspondence to noisy normals. Developing noise-robust variants, especially for scanned data, is an important direction for future work.

A different strategy may be to estimate multiple tangents per point using local mixture-of-subspaces methods such as generalized PCA~\cite{Vidal2005} or polyvector fields~\cite{diamanti2014designing, bessmeltsev2019vectorization}. Each estimated tangent would then generate a separate lifted copy of the same spatial sample, yielding a discrete approximation of the exceptional fiber. The blow-up itself is agnostic to how these tangent planes are obtained.

\begin{figure*}[tbp]
  \centering
  \includegraphics[width=\linewidth]{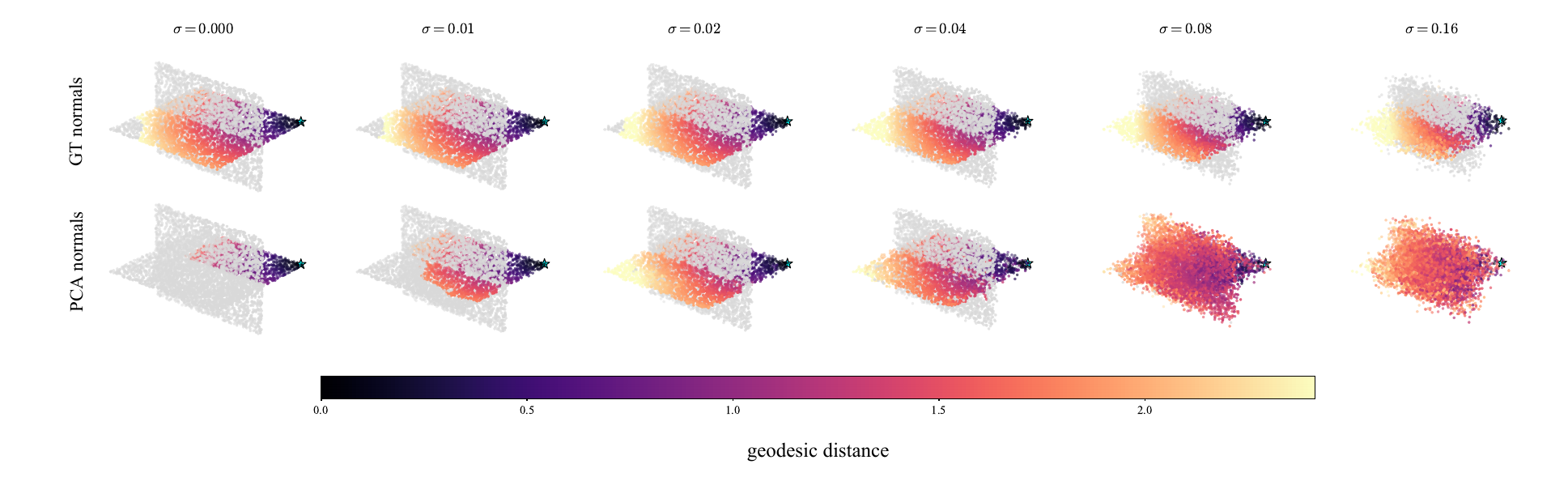}
  \caption{\label{fig:tangent-noise}Geodesic distances on two intersecting planes ($N{=}6000$, $k{=}20$, $\alpha{=}1.0$, $\sigma_x{=}\sigma_u{=}0.5$). Top row: ground-truth tangent frames; bottom row: PCA-estimated tangents ($k_{\mathrm{pca}}{=}30$). Columns show increasing positional noise $\sigma \in \{0, 0.01, 0.02, 0.04, 0.08, 0.16\}$. Grey points are unreachable (infinite geodesic distance), indicating correct sheet separation. The lifted heat method is robust to positional noise, but degrades when tangent estimates are corrupted, as PCA normals near the intersection line conflate intersecting components.}
\end{figure*}

Estimating multiple tangents could also improve the resolution of tangential intersections.
The first lift can convert a tangential intersection into a transverse one, with sheets meeting at an angle determined by their curvature difference. 
Lifted $k$-NN neighborhoods may still include some cross-sheet samples near the singular point. In practice, this contamination is localized to a small neighborhood of the singularity that shrinks as the curvature difference between sheets grows. 
Figure~\ref{fig:teaser} illustrates this, where the second lift achieves good separation despite the tangential contact. 
A more robust approach would involve handling multiple tangents at each blow-up level, which we leave to future work.
Similarly, while the blow-up can be iterated, our experiments mainly validate the first lift and simple second-level examples. A systematic study of higher-order lifts on complex singularities is left to future work.

A second limitation is the growth of ambient dimension under iteration. For surfaces in $\mathbb{R}^3$, the first lift lies in $\mathbb{R}^{12}$ and the second in $\mathbb{R}^{156}$. Although the intrinsic dimension remains unchanged, the projector representation is redundant, since symmetric rank-$d$ projectors have far fewer degrees of freedom than the ambient space. This overhead was still manageable in our experiments, but it appears in the timings. On the $200{,}000$-point \emph{Glykon} example in Figure \ref{fig:teaser}, constructing the level-one and level-two lifts took $38.1\,\mathrm{s}$ and $76.9\,\mathrm{s}$, respectively, while assembling the lifted affinity graph and linear systems for the lifted heat-method increased from about $4.4\,\mathrm{s}$ at level 1 to $30.7\,\mathrm{s}$ at level 2 (all timings measured on an Intel Core i9-9900K CPU with 8 physical cores and 16 logical cores, 48\,GiB of RAM, and an NVIDIA GeForce RTX 3070). In particular, exact kD-tree queries in the lifted ambient coordinates are likely to become a bottleneck as the dimension grows, since their effectiveness degrades in higher dimensions. For higher-order iterations, more compact Grassmannian coordinates, localized lifting near singular regions, or approximate nearest-neighbor methods may therefore be necessary.
For oriented hypersurfaces, a cheaper alternative is to lift using the normal itself.
This avoids the $\mathbb{R}^{n+n^2}$ projector embedding and distinguishes opposite normal orientations. We use projectors because they are orientation-free and extend directly to arbitrary codimension via $\mathrm{Gr}(d,n)$. However, if reliable oriented normals are available, the oriented lift can be substituted.

Beyond the operators explored here, the discrete tangent blow-up may also be useful for tasks such as interpolation and remeshing in the presence of singularities. Extending the framework to mixed-dimensional singularities, and clarifying its connection to higher-order jet constructions, are also natural next steps.\label{sec:discussion}
    
\section{Conclusion}
We introduced the discrete tangent blow-up, a lifted representation for point clouds with non-manifold geometry. By lifting each sample and its tangent plane to $\mathbb{R}^n \times \mathrm{Gr}(d,n)$, the construction separates incident sheets in a metric space where standard Euclidean structures can be applied without combinatorial repair or explicit singularity detection. The projector formulation yields a simple embedding, supports iteration to incorporate higher-order geometric information, and leads to discretized differential operators that can be defined directly in the lifted domain.
We view this work as a step toward geometry processing methods that treat singularities as part of the data.
We hope this perspective opens the door to a broader class of geometry processing algorithms that operate directly on non-manifold geometry.

\begin{figure}
  \centering
  \includegraphics[width=\linewidth]{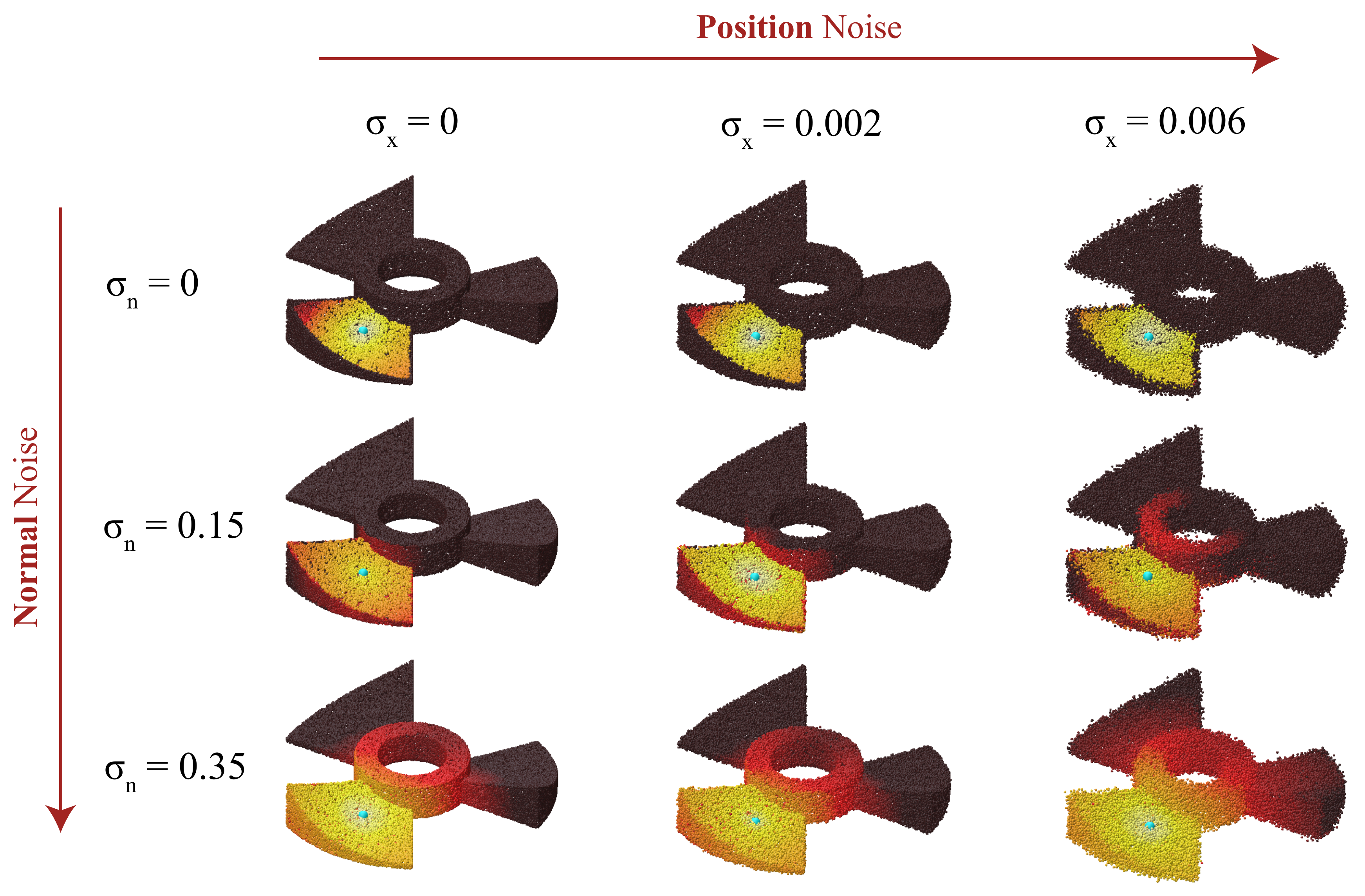}
  \caption{\label{fig:noise-study}Effect of noise on lifted heat diffusion. Columns add increasing Gaussian noise to point positions, reported as fractions of the bounding-box diagonal. Rows add increasing Gaussian noise to normals. Normals are estimated by local PCA from the original point cloud with $k=10$, and the lift uses $\alpha=2.0$. Moderate perturbations preserve the qualitative diffusion pattern, while high combined noise causes visible heat leakage.}
\end{figure}
    

\section*{Acknowledgements}
The MIT Geometric Data Processing Group acknowledges the generous support of National Science Foundation grants IIS2335492 and OAC2403239, from the CSAIL Future of Data and FinTechAI programs, from the MIT--IBM Watson AI Laboratory, from the Wistron Corporation, from the MIT Generative AI Impact Consortium, from the Toyota--CSAIL Joint Research Center, from the Natural Sciences and Engineering Research Council of Canada, and from Schmidt Sciences. 
Additional support was provided by SideFX Software.
Finally, we thank Vidit Nanda for helpful discussions.


\printbibliography

@String{tog = "ACM TOG"}

@article{nash1996arc,
  title={Arc structure of singularities},
  author={Nash Jr, John F},
  journal={Duke Math. J},
  volume={81},
  number={1},
  pages={31--38},
  year={1996}
}

@article{Semple1954,
  title={Some investigations in the geometry of curve and surface elements},
  author={Semple, J Greenlees},
  journal={Proceedings of the London Mathematical Society},
  volume={3},
  number={1},
  pages={24--49},
  year={1954},
  publisher={Wiley Online Library}
}

@article{Nobile1975,
  title={Some properties of the Nash blowing-up},
  author={Nobile, Augusto},
  journal={Pacific Journal of Mathematics},
  volume={60},
  number={1},
  pages={297--305},
  year={1975},
  publisher={Mathematical Sciences Publishers}
}

@inproceedings{GonzalezSprinberg1977,
  title={{\'E}ventails en dimension 2 et transform{\'e} de Nash},
  author={Gonz{\'a}lez-Sprinberg, Gerardo},
  journal={(No Title)},
  year={1977}
}

@incollection{Hironaka1983,
  title={On Nash blowing-up},
  author={Hironaka, Heisuke},
  booktitle={Arithmetic and Geometry: Papers Dedicated to IR Shafarevich on the Occasion of His Sixtieth Birthday. Volume II: Geometry},
  pages={103--111},
  year={1983},
  publisher={Springer}
}

@phdthesis{Rebassoo1977,
  author    = {Rebassoo, Vello},
  title     = {Desingularisation properties of the {N}ash blowing-up process},
  school    = {University of Washington},
  year      = {1977},
}

@article{CastilloEtAl2026,
  title={Nash blowup fails to resolve singularities in dimensions four and higher},
  author={Castillo, Federico and Duarte, Daniel and Leyton-{\'A}lvarez, Maximiliano and Liendo, Alvaro},
  journal={Annals of Mathematics},
  volume={203},
  number={2},
  pages={677--694},
  year={2026},
  publisher={Department of Mathematics, Princeton University Princeton, New Jersey, USA}
}

@article{BuetLeonardiMasnou2017,
  title={A varifold approach to surface approximation},
  author={Buet, Blanche and Leonardi, Gian Paolo and Masnou, Simon},
  journal={Archive for Rational Mechanics and Analysis},
  volume={226},
  number={2},
  pages={639--694},
  year={2017},
  publisher={Springer}
}

@article{BuetLeonardiMasnou2022,
  title={Weak and approximate curvatures of a measure: a varifold perspective},
  author={Buet, Blanche and Leonardi, Gian Paolo and Masnou, Simon},
  journal={Nonlinear Analysis},
  volume={222},
  pages={112983},
  year={2022},
  publisher={Elsevier}
}

@article{BuetRumpf2022,
  title={Mean curvature motion of point cloud varifolds},
  author={Buet, Blanche and Rumpf, Martin},
  journal={ESAIM: Mathematical Modelling and Numerical Analysis},
  volume={56},
  number={5},
  pages={1773--1808},
  year={2022},
  publisher={EDP Sciences}
}

@inproceedings{PottmannEtAl2004,
  title={The isophotic metric and its application to feature sensitive morphology on surfaces},
  author={Pottmann, Helmut and Steiner, Tibor and Hofer, Michael and Haider, Christoph and Hanbury, Allan},
  booktitle={European Conference on Computer Vision},
  pages={560--572},
  year={2004},
  organization={Springer}
}

@article{LaiEtAl2006,
  title={Robust feature classification and editing},
  author={Lai, Yu-Kun and Zhou, Qian-Yi and Hu, Shi-Min and Wallner, Johannes and Pottmann, Helmut},
  journal={IEEE Transactions on Visualization and Computer Graphics},
  volume={13},
  number={1},
  pages={34--45},
  year={2006},
  publisher={IEEE}
}

@inproceedings{LevyBonneel2013,
  title={Variational anisotropic surface meshing with Voronoi parallel linear enumeration},
  author={L{\'e}vy, Bruno and Bonneel, Nicolas},
  booktitle={Proceedings of the 21st international meshing roundtable},
  pages={349--366},
  year={2013},
  publisher={Springer}
}

@inproceedings{CanasGortler2006,
  title={Surface remeshing in arbitrary codimensions},
  author={Canas, Guillermo D and Gortler, Steven J},
  journal={The Visual Computer},
  volume={22},
  number={9},
  pages={885--895},
  year={2006},
  publisher={Springer}
}

@inproceedings{TomasiManduchi1998,
  title={Bilateral filtering for gray and color images},
  author={Tomasi, Carlo and Manduchi, Roberto},
  booktitle={Sixth international conference on computer vision (IEEE Cat. No. 98CH36271)},
  pages={839--846},
  year={1998},
  organization={IEEE}
}

@article{Jones2003,
  title={Non-iterative, feature-preserving mesh smoothing},
  author={Jones, Thouis R and Durand, Fr{\'e}do and Desbrun, Mathieu},
  booktitle={ACM SIGGRAPH 2003 Papers},
  pages={943--949},
  year={2003}
}

@article{Fleishman2003,
  title={Bilateral mesh denoising},
  author={Fleishman, Shachar and Drori, Iddo and Cohen-Or, Daniel},
  booktitle={ACM SIGGRAPH 2003 Papers},
  pages={950--953},
  year={2003}
}

@inproceedings{AndreuxEtAl2014,
  title={Anisotropic Laplace-Beltrami operators for shape analysis},
  author={Andreux, Mathieu and Rodola, Emanuele and Aubry, Mathieu and Cremers, Daniel},
  booktitle={European conference on computer vision},
  pages={299--312},
  year={2014},
  organization={Springer}
}

@inproceedings{BoscainiEtAl2016,
  title={Anisotropic diffusion descriptors},
  author={Boscaini, Davide and Masci, Jonathan and Rodol{\`a}, Emanuele and Bronstein, Michael M and Cremers, Daniel},
  booktitle={Computer Graphics Forum},
  volume={35},
  number={2},
  pages={431--441},
  year={2016},
  organization={Wiley Online Library}
}

@article{tinarrage2023recovering,
  title={Recovering the homology of immersed manifolds},
  author={Tinarrage, Rapha{\"e}l},
  journal={Discrete \& Computational Geometry},
  volume={69},
  number={3},
  pages={659--744},
  year={2023},
  publisher={Springer}
}

@article{charton2021mesh,
  title={Mesh repairing using topology graphs},
  author={Charton, Jerome and Baek, Stephen and Kim, Youngjun},
  journal={Journal of Computational Design and Engineering},
  volume={8},
  number={1},
  pages={251--267},
  year={2021},
  publisher={Oxford University Press}
}

@inproceedings{ValqueLazard2025,
  title={Resolving self-intersections in 3D meshes while preserving floating-point coordinates},
  author={Valque, L{\'e}o and Lazard, Sylvain},
  booktitle={Computer Graphics Forum},
  volume={44},
  number={5},
  pages={e70197},
  year={2025},
  organization={Wiley Online Library}
}

@article{ZhouEtAl2016,
  title={Mesh arrangements for solid geometry},
  author={Zhou, Qingnan and Grinspun, Eitan and Zorin, Denis and Jacobson, Alec},
  journal={ACM Transactions on Graphics (TOG)},
  volume={35},
  number={4},
  pages={1--15},
  year={2016},
  publisher={ACM New York, NY, USA}
}

@inproceedings{von2023topological,
  title={Topological singularity detection at multiple scales},
  author={Von Rohrscheidt, Julius and Rieck, Bastian},
  booktitle={International Conference on Machine Learning},
  pages={35175--35197},
  year={2023},
  organization={PMLR}
}

@article{lim2025hades,
  title={Hades: Fast singularity detection with local measure comparison},
  author={Lim, Uzu and Oberhauser, Harald and Nanda, Vidit},
  journal={SIAM Journal on Mathematics of Data Science},
  volume={7},
  number={4},
  pages={1882--1903},
  year={2025},
  publisher={SIAM}
}

@inproceedings{sharp2020laplacian,
  title={A laplacian for nonmanifold triangle meshes},
  author={Sharp, Nicholas and Crane, Keenan},
  booktitle={Computer Graphics Forum},
  volume={39},
  number={5},
  pages={69--80},
  year={2020},
  organization={Wiley Online Library}
}

@article{attene2013polygon,
  title={Polygon mesh repairing: An application perspective},
  author={Attene, Marco and Campen, Marcel and Kobbelt, Leif},
  journal={ACM Computing Surveys (CSUR)},
  volume={45},
  number={2},
  pages={1--33},
  year={2013},
  publisher={ACM New York, NY, USA}
}

@article{CazalsPouget2005,
  title={Estimating differential quantities using polynomial fitting of osculating jets},
  author={Cazals, Fr{\'e}d{\'e}ric and Pouget, Marc},
  journal={Computer aided geometric design},
  volume={22},
  number={2},
  pages={121--146},
  year={2005},
  publisher={Elsevier}
}

@article{LachaudCoeurjolly2023,
  title={Lightweight curvature estimation on point clouds with randomized corrected curvature measures},
  author={Lachaud, J-O and Coeurjolly, David and Labart, C{\'e}line and Romon, Pascal and Thibert, Boris},
  booktitle={Computer Graphics Forum},
  volume={42},
  number={5},
  pages={e14910},
  year={2023},
  organization={Wiley Online Library}
}

@article{LachaudRomonThibert2022,
  title={Corrected curvature measures},
  author={Lachaud, Jacques-Olivier and Romon, Pascal and Thibert, Boris},
  journal={Discrete \& Computational Geometry},
  volume={68},
  number={2},
  pages={477--524},
  year={2022},
  publisher={Springer}
}

@inproceedings{Rusinkiewicz2004,
  title={Estimating curvatures and their derivatives on triangle meshes},
  author={Rusinkiewicz, Szymon},
  booktitle={Proceedings. 2nd International Symposium on 3D Data Processing, Visualization and Transmission, 2004. 3DPVT 2004.},
  pages={486--493},
  year={2004},
  organization={IEEE}
}

@article{cao2021efficient,
  title={Efficient Weingarten map and curvature estimation on manifolds},
  author={Cao, Yueqi and Li, Didong and Sun, Huafei and Assadi, Amir H and Zhang, Shiqiang},
  journal={Machine Learning},
  volume={110},
  number={6},
  pages={1319--1344},
  year={2021},
  publisher={Springer}
}

@inproceedings{bendich2010stratification,
  title={Stratification Learning through Homology Inference.},
  author={Bendich, Paul and Mukherjee, Sayan and Wang, Bei},
  booktitle={AAAI Fall Symposium: Manifold Learning and Its Applications},
  year={2010}
}

@incollection{goresky1988stratified,
  title={Stratified morse theory},
  author={Goresky, Mark and MacPherson, Robert},
  booktitle={Stratified Morse Theory},
  pages={3--22},
  year={1988},
  publisher={Springer}
}

@book{friedman2020singular,
  title={Singular intersection homology},
  author={Friedman, Greg},
  volume={33},
  year={2020},
  publisher={Cambridge University Press}
}

@article{Bendich2011,
  title={Persistent intersection homology},
  author={Bendich, Paul and Harer, John},
  journal={Foundations of Computational Mathematics},
  volume={11},
  number={3},
  pages={305--336},
  year={2011},
  publisher={Springer}
}

@article{Nanda2020,
  title={Local cohomology and stratification},
  author={Nanda, Vidit},
  journal={Foundations of Computational Mathematics},
  volume={20},
  number={2},
  pages={195--222},
  year={2020},
  publisher={Springer}
}

@article{Stolz2020,
  title={Geometric anomaly detection in data},
  author={Stolz, Bernadette J and Tanner, Jared and Harrington, Heather A and Nanda, Vidit},
  journal={Proceedings of the national academy of sciences},
  volume={117},
  number={33},
  pages={19664--19669},
  year={2020},
  publisher={National Academy of Sciences}
}

@article{Vidal2005,
  title={Generalized principal component analysis (GPCA)},
  author={Vidal, Rene and Ma, Yi and Sastry, Shankar},
  journal={IEEE transactions on pattern analysis and machine intelligence},
  volume={27},
  number={12},
  pages={1945--1959},
  year={2005},
  publisher={IEEE}
}

@inproceedings{Bendich2012,
  title={Local homology transfer and stratification learning},
  author={Bendich, Paul and Wang, Bei and Mukherjee, Sayan},
  booktitle={Proceedings of the twenty-third annual ACM-SIAM symposium on Discrete Algorithms},
  pages={1355--1370},
  year={2012},
  organization={SIAM}
}

@article{Haro2006,
  title={Stratification learning: Detecting mixed density and dimensionality in high dimensional point clouds},
  author={Haro, Gloria and Randall, Gregory and Sapiro, Guillermo},
  journal={Advances in Neural Information Processing Systems},
  volume={19},
  year={2006}
}

@article{aamari2024theory,
  title={A theory of stratification learning},
  author={Aamari, Eddie and Berenfeld, Cl{\'e}ment},
  journal={arXiv preprint arXiv:2405.20066},
  year={2024}
}

@article{marshall2008effects,
  title={Effects of nonplanar fault topology and mechanical interaction on fault-slip distributions in the Ventura Basin, California},
  author={Marshall, Scott T and Cooke, Michele L and Owen, Susan E},
  journal={Bulletin of the Seismological Society of America},
  volume={98},
  number={3},
  pages={1113--1127},
  year={2008},
  publisher={Seismological Society of America}
}

@article{edelman1998geometry,
  title={The geometry of algorithms with orthogonality constraints},
  author={Edelman, Alan and Arias, Tom{\'a}s A and Smith, Steven T},
  journal={SIAM journal on Matrix Analysis and Applications},
  volume={20},
  number={2},
  pages={303--353},
  year={1998},
  publisher={SIAM}
}

@book{golub2013matrix,
  title={Matrix computations},
  author={Golub, Gene H and Van Loan, Charles F},
  year={2013},
  publisher={JHU press}
}

@book{boumal2023introduction,
  title={An introduction to optimization on smooth manifolds},
  author={Boumal, Nicolas},
  year={2023},
  publisher={Cambridge University Press}
}

@article{bendokat2024grassmann,
  title={A Grassmann manifold handbook: Basic geometry and computational aspects},
  author={Bendokat, Thomas and Zimmermann, Ralf and Absil, P-A},
  journal={Advances in Computational Mathematics},
  volume={50},
  number={1},
  pages={6},
  year={2024},
  publisher={Springer}
}

@article{conway1996packing,
  title={Packing lines, planes, etc.: Packings in Grassmannian spaces},
  author={Conway, John H and Hardin, Ronald H and Sloane, Neil JA},
  journal={Experimental mathematics},
  volume={5},
  number={2},
  pages={139--159},
  year={1996},
  publisher={Taylor \& Francis}
}

@article{Thingi10K,
  title={Thingi10K: A Dataset of 10,000 3D-Printing Models},
  author={Zhou, Qingnan and Jacobson, Alec},
  journal={arXiv preprint arXiv:1605.04797},
  year={2016}
}

@misc{threedscans,
  author       = {Laric, Oliver},
  title        = {Three D Scans},
  howpublished = {\url{http://threedscans.com/}},
  note         = {Accessed: 2026-04},
  year         = {2012}
}

@inproceedings{hua2016scenenn,
  title={Scenenn: A scene meshes dataset with annotations},
  author={Hua, Binh-Son and Pham, Quang-Hieu and Nguyen, Duc Thanh and Tran, Minh-Khoi and Yu, Lap-Fai and Yeung, Sai-Kit},
  booktitle={2016 fourth international conference on 3D vision (3DV)},
  pages={92--101},
  year={2016},
  organization={Ieee}
}

@article{belkin2006convergence,
  title={Convergence of Laplacian eigenmaps},
  author={Belkin, Mikhail and Niyogi, Partha},
  journal={Advances in neural information processing systems},
  volume={19},
  year={2006}
}

@inproceedings{belkin2009constructing,
  title={Constructing Laplace operator from point clouds in Rd},
  author={Belkin, Mikhail and Sun, Jian and Wang, Yusu},
  booktitle={Proceedings of the twentieth annual ACM-SIAM symposium on Discrete algorithms},
  pages={1031--1040},
  year={2009},
  organization={SIAM}
}

@inproceedings{hoppe1992surface,
  title={Surface reconstruction from unorganized points},
  author={Hoppe, Hugues and DeRose, Tony and Duchamp, Tom and McDonald, John and Stuetzle, Werner},
  booktitle={Proceedings of the 19th annual conference on computer graphics and interactive techniques},
  pages={71--78},
  year={1992}
}

@book{robbin2022introduction,
  title={Introduction to differential geometry},
  author={Robbin, Joel W and Salamon, Dietmar A},
  year={2022},
  publisher={Springer Nature}
}

@inproceedings{vaxman2016directional,
  title={Directional field synthesis, design, and processing},
  author={Vaxman, Amir and Campen, Marcel and Diamanti, Olga and Panozzo, Daniele and Bommes, David and Hildebrandt, Klaus and Ben-Chen, Mirela},
  booktitle={Computer graphics forum},
  volume={35},
  number={2},
  pages={545--572},
  year={2016},
  organization={Wiley Online Library}
}

@incollection{lee2003smooth,
  title={Smooth manifolds},
  author={Lee, John M},
  booktitle={Introduction to smooth manifolds},
  pages={1--29},
  year={2003},
  publisher={Springer}
}

@article{coifman2006diffusion,
  title={Diffusion maps},
  author={Coifman, Ronald R and Lafon, St{\'e}phane},
  journal={Applied and computational harmonic analysis},
  volume={21},
  number={1},
  pages={5--30},
  year={2006},
  publisher={Elsevier}
}

@article{zelnik2004self,
  title={Self-tuning spectral clustering},
  author={Zelnik-Manor, Lihi and Perona, Pietro},
  journal={Advances in neural information processing systems},
  volume={17},
  year={2004}
}

@article{liang2013solving,
  title={Solving partial differential equations on point clouds},
  author={Liang, Jian and Zhao, Hongkai},
  journal={SIAM Journal on Scientific Computing},
  volume={35},
  number={3},
  pages={A1461--A1486},
  year={2013},
  publisher={SIAM}
}

@article{wiersma2022deltaconv,
  title={Deltaconv: anisotropic operators for geometric deep learning on point clouds},
  author={Wiersma, Ruben and Nasikun, Ahmad and Eisemann, Elmar and Hildebrandt, Klaus},
  journal={ACM Transactions on Graphics (ToG)},
  volume={41},
  number={4},
  pages={1--10},
  year={2022},
  publisher={ACM New York, NY, USA}
}

@article{nealen2004short,
  title={An as-short-as-possible introduction to the least squares, weighted least squares and moving least squares methods for scattered data approximation and interpolation},
  author={Nealen, Andrew},
  journal={URL: http://www. nealen. com/projects},
  volume={130},
  number={150},
  pages={25},
  year={2004}
}

@article{crane2013geodesics,
  title={Geodesics in heat: A new approach to computing distance based on heat flow},
  author={Crane, Keenan and Weischedel, Clarisse and Wardetzky, Max},
  journal={ACM Transactions on Graphics (ToG)},
  volume={32},
  number={5},
  pages={1--11},
  year={2013},
  publisher={ACM New York, NY, USA}
}

@article{von2007tutorial,
  title={A tutorial on spectral clustering},
  author={Von Luxburg, Ulrike},
  journal={Statistics and computing},
  volume={17},
  number={4},
  pages={395--416},
  year={2007},
  publisher={Springer}
}

@article{ng2001spectral,
  title={On spectral clustering: Analysis and an algorithm},
  author={Ng, Andrew and Jordan, Michael and Weiss, Yair},
  journal={Advances in neural information processing systems},
  volume={14},
  year={2001}
}

@inproceedings{ester1996density,
  title={A density-based algorithm for discovering clusters in large spatial databases with noise},
  author={Ester, Martin and Kriegel, Hans-Peter and Sander, J{\"o}rg and Xu, Xiaowei and others},
  booktitle={kdd},
  volume={96},
  number={34},
  pages={226--231},
  year={1996}
}

@article{belkin2003laplacian,
  title={Laplacian eigenmaps for dimensionality reduction and data representation},
  author={Belkin, Mikhail and Niyogi, Partha},
  journal={Neural computation},
  volume={15},
  number={6},
  pages={1373--1396},
  year={2003},
  publisher={MIT Press}
}

@book{do2016differential,
  title={Differential geometry of curves and surfaces: revised and updated second edition},
  author={Do Carmo, Manfredo P},
  year={2016},
  publisher={Courier Dover Publications}
}

@book{abbena2017modern,
  title={Modern differential geometry of curves and surfaces with Mathematica},
  author={Abbena, Elsa and Salamon, Simon and Gray, Alfred},
  year={2017},
  publisher={Chapman and Hall/CRC}
}

@article{petit1981representation,
  title={Une repr{\'e}sentation analytique de la surface de Boy},
  author={Petit, Jean-Pierre and Souriau, J},
  journal={Compte Rendus de l’Acad{\'e}mie des Sciences de Paris},
  volume={293},
  pages={5},
  year={1981}
}

@misc{Busser-nonorientable,
  author       = {Busser, Alain},
  title        = {Non-Orientable Surfaces: Parametric Representations},
  year         = {2012},
  howpublished = {\url{https://eumat.sourceforge.net/Programs/Examples/Non-Orientable\%20Surfaces.html}},
  note         = {Accessed April 2026}
}

@article{banchoff1976minimal,
  title={Minimal submanifolds of the bicylinder boundary},
  author={Banchoff, Thomas F},
  journal={Boletim da Sociedade Brasileira de Matem{\'a}tica-Bulletin/Brazilian Mathematical Society},
  volume={7},
  number={1},
  pages={37--57},
  year={1976},
  publisher={Springer}
}

@inproceedings{hoppe1994piecewise,
  title={Piecewise smooth surface reconstruction},
  author={Hoppe, Hugues and DeRose, Tony and Duchamp, Tom and Halstead, Mark and Jin, Hubert and McDonald, John and Schweitzer, Jean and Stuetzle, Werner},
  booktitle={Proceedings of the 21st annual conference on Computer graphics and interactive techniques},
  pages={295--302},
  year={1994}
}

@article{bessmeltsev2019vectorization,
  title={Vectorization of line drawings via polyvector fields},
  author={Bessmeltsev, Mikhail and Solomon, Justin},
  journal={ACM Transactions on Graphics (TOG)},
  volume={38},
  number={1},
  pages={1--12},
  year={2019},
  publisher={ACM New York, NY, USA}
}

@inproceedings{diamanti2014designing,
  title={Designing N-PolyVector fields with complex polynomials},
  author={Diamanti, Olga and Vaxman, Amir and Panozzo, Daniele and Sorkine-Hornung, Olga},
  booktitle={Computer Graphics Forum},
  volume={33},
  number={5},
  pages={1--11},
  year={2014},
  organization={Wiley Online Library}
}

@article{kambhatla1997dimension,
  title={Dimension reduction by local principal component analysis},
  author={Kambhatla, Nandakishore and Leen, Todd K},
  journal={Neural computation},
  volume={9},
  number={7},
  pages={1493--1516},
  year={1997},
  publisher={MIT Press One Rogers Street, Cambridge, MA 02142-1209, USA journals-info~…}
}

@article{li2022hsurf,
  title={HSurf-Net: Normal estimation for 3D point clouds by learning hyper surfaces},
  author={Li, Qing and Liu, Yu-Shen and Cheng, Jin-San and Wang, Cheng and Fang, Yi and Han, Zhizhong},
  journal={Advances in Neural Information Processing Systems},
  volume={35},
  pages={4218--4230},
  year={2022}
}

@inproceedings{li2025high,
  title={High-quality Point Cloud Oriented Normal Estimation via Hybrid Angular and Euclidean Distance Encoding},
  author={Li, Yuanqi and Huang, Jingcheng and Wang, Hongshen and Lv, Peiyuan and Liu, Yansong and Zheng, Jiuming and Guo, Jie and Guo, Yanwen},
  booktitle={Proceedings of the Computer Vision and Pattern Recognition Conference},
  pages={1287--1296},
  year={2025}
}

@inproceedings{kohlbrenner2025symmetrized,
  title={Symmetrized Poisson Reconstruction},
  author={Kohlbrenner, Maximilian and Liu, Hongyi and Alexa, Marc and Kazhdan, Misha},
  booktitle={Computer Graphics Forum},
  volume={44},
  number={5},
  pages={e70210},
  year={2025},
  organization={Wiley Online Library}
}

@inproceedings{yang2025edgemovingnet,
  title={EdgeMovingNet: Edge-preserving Point Cloud Reconstruction via Joint Geometry Features},
  author={Yang, Xinran and Ji, Donghao and Li, Yuanqi and Xie, Junyuan and Guo, Jie and Guo, Yanwen},
  booktitle={Proceedings of the Computer Vision and Pattern Recognition Conference},
  pages={22150--22160},
  year={2025}
}

@inproceedings{arnal2026survey,
  title={Survey on differential estimators for 3d point clouds},
  author={Arnal--Anger, L{\'e}o and Lejemble, Thibault and Coeurjolly, David and Barthe, Lo{\"\i}c and Mellado, Nicolas},
  booktitle={Computer Graphics Forum},
  pages={e70394},
  year={2026},
  organization={Wiley Online Library}
}

@inproceedings{erler2020points2surf,
  title={Points2surf learning implicit surfaces from point clouds},
  author={Erler, Philipp and Guerrero, Paul and Ohrhallinger, Stefan and Mitra, Niloy J and Wimmer, Michael},
  booktitle={European conference on computer vision},
  pages={108--124},
  year={2020},
  organization={Springer}
}

@inproceedings{ma2021neural,
  title={Neural-Pull: Learning Signed Distance Function from Point clouds by Learning to Pull Space onto Surface},
  author={Ma, Baorui and Han, Zhizhong and Liu, Yu-Shen and Zwicker, Matthias},
  booktitle={International Conference on Machine Learning},
  pages={7246--7257},
  year={2021},
  organization={PMLR}
}

@inproceedings{guerrero2018pcpnet,
  title={Pcpnet learning local shape properties from raw point clouds},
  author={Guerrero, Paul and Kleiman, Yanir and Ovsjanikov, Maks and Mitra, Niloy J},
  booktitle={Computer graphics forum},
  volume={37},
  number={2},
  pages={75--85},
  year={2018},
  organization={Wiley Online Library}
}

@inproceedings{ben2020deepfit,
  title={DeepFit: 3D surface fitting via neural network weighted least squares},
  author={Ben-Shabat, Yizhak and Gould, Stephen},
  booktitle={European conference on computer vision},
  pages={20--34},
  year={2020},
  organization={Springer}
}

\section*{Appendix A: Proofs of Separation and Smoothness Theorems}

\begin{theorem}[Lifted Separation]\label{sec:separation-proof} 
    Let $M_1, M_2 \subset \mathbb{R}^n$ be smooth $d$-dimensional
    submanifolds intersecting at a point $x_0$ with distinct tangent
    planes. Write $P_{0}^{(1)}, P_{0}^{(2)}$ for their respective
    tangent projectors at $x_0$, and let
    \[
      \delta \;=\;
        d_{\mathrm{chord}}\!\bigl(T_{x_0}M_1,\;T_{x_0}M_2\bigr)
      \;>\; 0
    \]
    be the chordal distance between the two tangent planes. Then
    there exists $\varepsilon_0 > 0$ (depending on~$\delta$ and
    the second fundamental forms of $M_1, M_2$ near~$x_0$) such
    that for all $p \in M_1$, $q \in M_2$ with
    $\|p - x_0\|,\,\|q - x_0\| < \varepsilon_0$,
    \[
      d_{\mathcal{L}}\!\bigl(\tilde{p},\,\tilde{q}\bigr)
      \;\geq\;
      \frac{\sqrt{\alpha}\;\delta}{\sqrt{2}}
      \;>\; 0,
    \]
    where $\tilde{p} = (p,\,T_pM_1)$, $\tilde{q} = (q,\,T_qM_2)$
    are the lifted points and~$d_{\mathcal{L}}$ is the lifted
    metric.
\end{theorem}

\begin{proof}
Write $P_x^{(k)}$ for the tangent projector of $M_k$ at
$x \in M_k$, and $P_0^{(k)} = P_{x_0}^{(k)}$ for its value
at $x_0$, for $k = 1, 2$.

\medskip\noindent\textbf{Step 1 (Projector Lipschitz bound).}
Since each $M_k$ is smooth, the projector map
$x \mapsto P_x^{(k)}$ is smooth on $M_k$. 
Its directional derivative along a tangent vector $v$ is
the projector derivative $\mathcal{P}_v^{(k)}$, 
which decomposes as
$\mathcal{P}_v^{(k)} = N^{(k)} B^{(k)}(v)\,(U^{(k)})^\top
+ U^{(k)} B^{(k)}(v)^\top (N^{(k)})^\top$.
The two summands have orthogonal column spaces
($\mathrm{col}(N^{(k)})$ and $\mathrm{col}(U^{(k)})$),
so we can use $\|A + B\|_F^2 = \|A\|_F^2 + \|B\|_F^2$
(the trace $\mathrm{tr}(A^\top B) = 0$).
Multiplying by matrices with orthonormal columns preserves
the Frobenius norm, so each summand contributes
$\|B^{(k)}(v)\|_F$, giving
\[
  \|\mathcal{P}_v^{(k)}\|_F^2
  \;=\;
  2\,\|B^{(k)}(v)\|_F^2
  \;\leq\;
  2\,L_k^2\,\|v\|^2,
\]
where
$L_k = \max\{\|B_x^{(k)}(v)\|_F : x \in M_k,\;
\|x - x_0\| \leq \varepsilon_0,\;\|v\| = 1\}$
is finite by smoothness and compactness.
By the mean value inequality,
\[
  \bigl\|P_x^{(k)} - P_0^{(k)}\bigr\|_F
  \;\leq\;
  \sqrt{2}\;L_k\;d_{M_k}(x, x_0),
\]
where $d_{M_k}$ is the geodesic distance on $M_k$.
For $\varepsilon_0$ sufficiently small,
$d_{M_k}(x, x_0) \leq 2\,\|x - x_0\|$, so setting
$C_k = 2\sqrt{2}\,L_k$ gives
$\|P_x^{(k)} - P_0^{(k)}\|_F \leq C_k\,\|x - x_0\|$.

\medskip\noindent\textbf{Step 2 (Triangle inequality on
projectors).} For $p \in M_1$ and $q \in M_2$ with
$\|p - x_0\|,\,\|q - x_0\| < \varepsilon_0$, set
$C = C_1 + C_2$:
\begin{align*}
  \bigl\|P_p^{(1)} - P_q^{(2)}\bigr\|_F
  &\geq \bigl\|P_0^{(1)} - P_0^{(2)}\bigr\|_F
    - \bigl\|P_p^{(1)} - P_0^{(1)}\bigr\|_F
    - \bigl\|P_q^{(2)} - P_0^{(2)}\bigr\|_F \\
  &\geq \sqrt{2}\,\delta - C\,\varepsilon_0,
\end{align*}
where we used
$\|P_0^{(1)} - P_0^{(2)}\|_F = \sqrt{2}\,\delta$,
the standard identity relating the projector Frobenius
distance to the chordal distance on the Grassmannian.

\medskip\noindent\textbf{Step 3 (Lifted distance).}
The squared lifted distance satisfies
\[
  d_{\mathcal{L}}^2(\tilde{p},\,\tilde{q})
  \;=\; \|p - q\|^2
    + \tfrac{\alpha}{2}
      \bigl\|P_p^{(1)} - P_q^{(2)}\bigr\|_F^2
  \;\geq\;
  \tfrac{\alpha}{2}
    \bigl(\sqrt{2}\,\delta - C\,\varepsilon_0\bigr)^2.
\]
Choosing
$\varepsilon_0 < (\sqrt{2} - 1)\,\delta\,/\,C$
ensures $\sqrt{2}\,\delta - C\,\varepsilon_0 > \delta$,
so that
\[
  d_{\mathcal{L}}(\tilde{p},\,\tilde{q})
  \;\geq\;
  \sqrt{\tfrac{\alpha}{2}}\;\delta
  \;=\;
  \frac{\sqrt{\alpha}\;\delta}{\sqrt{2}}.
\]
\end{proof}

\begin{theorem}[Regularity of the Lifted Manifold]\label{sec:regularity-proof}
    Let $M \subset \mathbb{R}^n$ be a $C^r$ submanifold of dimension~$d$ with $r \geq 2$.  Then the lifted embedding $\Phi : M \to \mathbb{R}^{n+n^2}$, $\Phi(x) = \bigl(x,\,\sqrt{\alpha/2}\;\mathrm{vec}(P_x)\bigr)$, is a $C^{r-1}$ embedding, and the lifted manifold $\widehat{M} = \Phi(M)$ is a $C^{r-1}$ submanifold of dimension~$d$.
\end{theorem}

\begin{proof}
The map $\Phi$ is a section of the projection 
$\pi : \mathbb{R}^{n+n^2} \to \mathbb{R}^n$ onto the first factor: 
by construction, $\pi \circ \Phi = \mathrm{id}_M$. For any section, 
injectivity ($\Phi(p) = \Phi(q) \Rightarrow p = \pi(\Phi(p)) = 
\pi(\Phi(q)) = q$), immersivity ($d\pi \circ d\Phi = \mathrm{id}$ 
implies $d\Phi$ is injective), and topological embedding 
($\Phi^{-1} = \pi|_{\widehat{M}}$ is continuous) are automatic. It 
therefore suffices to show that $\Phi$ is $C^{r-1}$.

The first component of $\Phi$ is the inclusion 
$M \hookrightarrow \mathbb{R}^n$, which is $C^r$. For the 
second component, let $\varphi : V \to M$ be a $C^r$ local 
parametrization with Jacobian 
$J(x) = D\varphi_{\varphi^{-1}(x)} \in \mathbb{R}^{n \times d}$. 
The projector admits the parametrization-free expression
\[
P_x = J(x)\,\bigl(J(x)^\top J(x)\bigr)^{-1}\, J(x)^\top.
\]
Since $\varphi$ is $C^r$, the Jacobian $J$ is $C^{r-1}$. The Gram 
matrix $J^\top J$ is $C^{r-1}$ and invertible (since $J$ has full 
column rank), and matrix inversion is smooth on the open set of 
invertible matrices. Therefore $x \mapsto P_x$ is $C^{r-1}$, and 
hence so is $\Phi$.
\end{proof}

\end{document}